\documentclass[11pt,amssymb,amsfont,a4paper]{article}
\usepackage{latexsym,amssymb,amsmath,amsthm,color}
\usepackage{geometry}
\usepackage{url}
\usepackage{graphicx}
\usepackage[utf8]{inputenc}
\usepackage{authblk}
\usepackage{multirow}

\theoremstyle{definition}
\newtheorem{definition}{Definition}
\newtheorem{construction}{Construction}

\theoremstyle{plain}
\newtheorem{theorem}{Theorem}
\newtheorem{proposition}[definition]{Proposition}
\newtheorem{lemma}[definition]{Lemma}

\newtheorem{corollary}[definition]{Corollary}

\title{A general family of MSRD codes and PMDS codes with smaller field sizes from extended Moore matrices}
\author{Umberto Mart{\'i}nez-Pe\~{n}as \thanks{umberto.martinez@uva.es}}
\affil{IMUVa-Mathematics Research Institute,\\University of Valladolid, Spain}

\date{}

\begin{document}

\maketitle

\begin{abstract}
We construct six new explicit families of linear maximum sum-rank distance (MSRD) codes, each of which has the smallest field sizes among all known MSRD codes for some parameter regime. Using them and a previous result of the author, we provide two new explicit families of linear partial MDS (PMDS) codes with smaller field sizes than previous PMDS codes for some parameter regimes. Our approach is to characterize evaluation points that turn extended Moore matrices into the parity-check matrix of a linear MSRD code. We then produce such sequences from codes with good Hamming-metric parameters. The six new families of linear MSRD codes with smaller field sizes are obtained using MDS codes, Hamming codes, BCH codes and three Algebraic-Geometry codes. The MSRD codes based on Hamming codes, of minimum sum-rank distance $ 3 $, meet a recent bound by Byrne et al. 


\end{abstract}

\section{Introduction} \label{sec intro}

\textit{Maximum distance separable} (MDS) codes are optimal in the sense that their minimum Hamming distance \cite{hamming} attains the Singleton bound \cite{singleton}, which is independent of the alphabet size (i.e., field size). 
Thus, in erasure scenarios where alphabets need not be too small and for fixed block lengths, MDS codes offer the best erasure correction capability. One of such erasure scenarios is node repair in distributed storage. 
However, repairing a single node out of $ n $ nodes using an MDS code of rate $ k/n $ requires contacting $ k $ other nodes. Thus repairing a single node results in a high latency when using MDS codes. \textit{Locally repairable codes} (LRCs) \cite{gopalan, kamath} may repair one node (or more generally, $ \delta -1 $ nodes) by contacting a small number $ r $ (called \textit{locality}) of other nodes. Simultaneously, they can correct many global erasures in catastrophic cases.
%

\textit{Partial MDS} (PMDS) codes \cite{blaum-RAID, gopalan-MR} are LRCs that can correct all the erasure patterns correctable by any other LRC over any alphabet but with the same information rate and locality constraints, if we assume that local repair sets are pair-wise non-intersecting (note that general LRCs do not require local repair sets to be non-intersecting). Singleton bounds for the global minimum distance of LRCs were given in \cite[Eq. (2)]{gopalan} and \cite[Th. 2.1]{kamath}. Any PMDS code attains such Singleton bounds, but not all LRCs attaining such bounds are PMDS. 

Several constructions of PMDS codes exist in the literature \cite{blaum-RAID, blaum-twoparities, neri-PMDS, cai-field, calis, gabrys, gopalan-MR, guruswami-PMDS, hu, universal-lrc, neri}. In Construction 1 in \cite{universal-lrc}, it was shown that any \textit{maximum sum-rank distance} (MSRD) code \cite{linearizedRS} may be explicitly turned into a PMDS code \cite[Th. 2]{universal-lrc}. Moreover, this Construction 1 enjoys further flexibility properties, such as enabling hierarchical PMDS codes (see \cite{universal-lrc}). As another application of such a flexibility, optimal LRCs with multiple disjoint repair sets were obtained based on MSRD codes in \cite{cai-disjoint}.

Apart from being used as PMDS codes for distributed storage \cite{universal-lrc}, MSRD codes have applications in reliable and secure multishot network coding \cite{multishot, secure-multishot}, rate-diversity optimal space-time codes with multiple fading blocks \cite{space-time-kumar, mohannad}, multilayer crisscross error correction \cite{multilayer}, and private information retrieval from locally repairable databases \cite{pirlrc}. 

Codes over small fields are preferable, as they enjoy lower computational complexity for encoding and decoding. In contrast with MDS codes, PMDS codes and MSRD codes with linear field sizes in the code length do not exist for arbitrary dimensions \cite{gopi, alberto-fundamental}. 
%
The problems of finding the smallest possible field sizes of PMDS codes and MSRD codes constitute two generalizations of the MDS conjecture. However, they are significantly harder since even the possible asymptotic field sizes are unknown for PMDS and MSRD codes, whereas it was known since \cite{reed-solomon, singleton} that MDS codes exist if, and only if, the field sizes grow at least linearly in the code length.

In this work, we obtain a general family of MSRD codes that extends linearized Reed-Solomon codes \cite{linearizedRS}, but also includes six new families of explicit MSRD codes (rows 2 to 7 in Table \ref{table comparisons MSRD}), each of which attains smaller field sizes than all other known MSRD codes for some parameter regime. For minimum sum-rank distance $ 3 $ (co-dimension $ 2 $), our MSRD codes meet a bound recently given in \cite[Th. 6.12]{alberto-fundamental}. We also obtain two new families of explicit PMDS codes (rows 2 and 3 in Table \ref{table comparisons PMDS}), each of which attains smaller field sizes than all other known PMDS codes for some parameter regime. See Section \ref{sec summary and final considerations} for a detailed summary and comparisons, and the Appendix for tables with concrete values of even field sizes and other parameters. 

The manuscript is organized as follows. In Section \ref{sec preliminaries}, we collect preliminaries on MDS, MSRD and PMDS codes. In Section \ref{sec conjugacy and ext moore matricex}, we characterize sequences of evaluation points that turn an extended Moore matrix into the parity-check matrix of an MSRD code. In Section \ref{sec technique of tensor products}, we construct such sequences via tensor products and a range of Hamming-metric codes. In Section \ref{sec summary and final considerations}, we provide a summary of the obtained MSRD and PMDS codes and compare their parameters among themselves and with known codes. 

\section{Preliminaries} \label{sec preliminaries}

We will denote $ \mathbb{N} = \{ 0,1,2, \ldots \} $ and $ \mathbb{N}_+ = \{ 1,2,3, \ldots \} $. For positive integers $ m \leq n $, we denote $ [n] = \{ 1,2, \ldots, n\} $ and $ [m,n] = \{ m, m+1, \ldots, n \} $. For a field $ \mathbb{F} $, we denote $ \mathbb{F}^* = \mathbb{F} \setminus \{ 0 \} $ and we use $ \langle \cdot \rangle_\mathbb{F} $ and $ \dim_\mathbb{F}(\cdot) $ to denote $ \mathbb{F} $-linear span and dimension over $ \mathbb{F} $, respectively. We denote by $ \mathbb{F}^{m \times n} $ the set of $ m \times n $ matrices with entries in $ \mathbb{F} $, and we denote $ \mathbb{F}^n = \mathbb{F}^{1 \times n} $. The group of invertible matrices in $ \mathbb{F}^{n \times n} $ is denoted by $ {\rm GL}_n(\mathbb{F}) $. A code in $ \mathbb{F}^n $ is any subset $ \mathcal{C} \subseteq \mathbb{F}^n $, and we say that $ \mathcal{C} $ is a linear code if it is an $ \mathbb{F} $-linear vector subspace of $ \mathbb{F}^n $. For matrices $ A_1, A_2, \ldots, $ $ A_g $ $ \in \mathbb{F}^{r \times s} $, for some positive integers $ g $, $ r $ and $ s $, we define the block-diagonal matrix
$$ {\rm diag} (A_1, A_2, \ldots, A_g) = \left( \begin{array}{cccc}
A_1 & 0 & \ldots & 0 \\
0 & A_2 & \ldots & 0 \\
\vdots & \vdots & \ddots & \vdots \\
0 & 0 & \ldots & A_g
\end{array} \right) \in \mathbb{F}^{gr \times gs} . $$
We will also denote by $ \mathbf{c} \cdot \mathbf{d} \in \mathbb{F} $ the conventional inner product of $ \mathbf{c}, \mathbf{d} \in \mathbb{F}^n $ (i.e., $ \mathbf{c} \cdot \mathbf{d} = \mathbf{c} \mathbf{d}^T $), and we denote the dual of a linear code $ \mathcal{C} \subseteq \mathbb{F}^n $ by $ \mathcal{C}^\perp = \{ \mathbf{d} \in \mathbb{F}^n \mid \mathbf{c} \cdot \mathbf{d} = 0, \textrm{ for all } \mathbf{c} \in \mathcal{C} \} \subseteq \mathbb{F}^n $.

For a prime power $ q $, we denote by $ \mathbb{F}_q $ the finite field with $ q $ elements. Throughout this manuscript, we will fix a prime power $ q $ and a finite-field extension $ \mathbb{F}_q \subseteq \mathbb{F}_{q^m} $, for some positive integer $ m $. The field $ \mathbb{F}_q $ will be called \textit{the base field} throughout the manuscript. Our target codes will be linear codes $ \mathcal{C} \subseteq \mathbb{F}_{q^m}^n $, hence we will usually call $ \mathbb{F}_{q^m} $ \textit{the field of linearity} of $ \mathcal{C} $.
%

\subsection{MDS codes} \label{subsec mds codes}

For a positive integer $ n $ and a field $ \mathbb{F} $, the \textit{Hamming weight} \cite{hamming} of a vector $ \mathbf{c} = (c_1, c_2, \ldots, $ $ c_n) $ $ \in \mathbb{F}^n $ is $ {\rm wt}_H(\mathbf{c}) = | \{ i \in [n] \mid c_i \neq 0 \} | $. We define the \textit{Hamming metric} $ {\rm d}_H : \left( \mathbb{F}^n \right) ^2 \longrightarrow \mathbb{N} $ by $ {\rm d}_H(\mathbf{c}, \mathbf{d}) = {\rm wt}_H(\mathbf{c} - \mathbf{d}) $, for all $ \mathbf{c}, \mathbf{d} \in \mathbb{F}^n $. For a (linear or non-linear) code $ \mathcal{C} \subseteq \mathbb{F}^n $, its \textit{minimum Hamming distance} is $ {\rm d}_H(\mathcal{C}) = \min \left\lbrace {\rm d}_H(\mathbf{c}, \mathbf{d}) \mid \mathbf{c}, \mathbf{d} \in \mathcal{C}, \mathbf{c} \neq \mathbf{d} \right\rbrace $. 

We next revisit the \textit{Singleton bound} and \textit{MDS codes} \cite{singleton}.

\begin{proposition} [\textbf{\cite{singleton}}] \label{prop singleton bound}
For any (linear or non-linear) code $ \mathcal{C} \subseteq \mathbb{F}^n $, it holds that $ | \mathcal{C} | \leq |\mathbb{F}|^{n - {\rm d}_H(\mathcal{C}) + 1} $. If equality holds, then we say that $ \mathcal{C} $ is a maximum distance separable (MDS) code.
\end{proposition}

\subsection{MSRD codes} \label{subsec msrd codes}

Fix positive integers $ m $ and $ r $, and an ordered basis $ \boldsymbol\alpha = ( \alpha_1, \alpha_2, \ldots, \alpha_m ) \in \mathbb{F}_{q^m}^m $ of $ \mathbb{F}_{q^m} $ over $ \mathbb{F}_q $. We define the \textit{matrix representation} map $ M_{\boldsymbol\alpha} : \mathbb{F}_{q^m}^r \longrightarrow \mathbb{F}_q^{m \times r} $ by 
\begin{equation}
M_{\boldsymbol\alpha} \left( \sum_{i=1}^m \alpha_i \mathbf{c}_i \right) = \left( \begin{array}{cccc}
c_{1,1} & c_{1,2} & \ldots & c_{1,r} \\
c_{2,1} & c_{2,2} & \ldots & c_{2,r} \\
\vdots & \vdots & \ddots & \vdots \\
c_{m,1} & c_{m,2} & \ldots & c_{m,r} \\
\end{array} \right) \in \mathbb{F}_q^{m \times r},
\label{eq def matrix representation map}
\end{equation}
where $ \mathbf{c}_i = (c_{i,1}, c_{i,2}, \ldots, c_{i,r}) \in \mathbb{F}_q^r $, for $ i = 1,2, \ldots, m $. In order to define sum-rank weights on vectors with components in $ \mathbb{F}_{q^m} $, we will subdivide them into subvectors as $ \mathbf{c} = (\mathbf{c}^{(1)}, $ $ \mathbf{c}^{(2)}, $ $ \ldots, $ $ \mathbf{c}^{(g)}) \in \mathbb{F}_{q^m}^{g r} $, where $ \mathbf{c}^{(i)} \in \mathbb{F}_{q^m}^r $, for $ i = 1,2, \ldots, g $, for a positive integer $ g $. Using (\ref{eq def matrix representation map}), we may consider $ \mathbf{c} \in \mathbb{F}_{q^m}^{gr} $ as a list of $ g $ matrices of size $ m \times r $ over $ \mathbb{F}_q $:
\begin{equation}
 \left( M_{\boldsymbol\alpha} \left( \mathbf{c}^{(1)} \right) , M_{\boldsymbol\alpha} \left( \mathbf{c}^{(2)} \right), \ldots, M_{\boldsymbol\alpha} \left( \mathbf{c}^{(g)} \right) \right) \in \left( \mathbb{F}_q^{m \times r} \right)^g .
\label{eq codeword as list of matrices}
\end{equation}

We now define the \textit{sum-rank metric}, which was explicitly defined in \cite[Sec. III-D]{multishot}, but previously used implicitly in \cite[Sec. III]{space-time-kumar}.

\begin{definition} [\textbf{Sum-rank metric \cite{space-time-kumar, multishot}}]
Let $ g $ be a positive integer, and let $ \mathbf{c} = (\mathbf{c}^{(1)}, $ $ \mathbf{c}^{(2)}, $ $ \ldots, $ $ \mathbf{c}^{(g)}) \in \mathbb{F}_{q^m}^{g r} $, where $ \mathbf{c}^{(i)} \in \mathbb{F}_{q^m}^r $, for $ i = 1,2, \ldots, g $. We define the \textit{sum-rank weight} of $ \mathbf{c} $, for the partition $ (g,r) $ over the base field $ \mathbb{F}_q $, by
$$ {\rm wt}_{SR}(\mathbf{c}) = \sum_{i=1}^g {\rm Rk} \left( M_{\boldsymbol\alpha} ( \mathbf{c}^{(i)} ) \right) . $$
We define the \textit{sum-rank metric} $ {\rm d}_{SR} : \left( \mathbb{F}_{q^m}^{g r} \right) ^2 \longrightarrow \mathbb{N} $, for the partition $ (g,r) $ over the base field $ \mathbb{F}_q $, by $ {\rm d}_{SR}(\mathbf{c}, \mathbf{d}) = {\rm wt}_{SR}(\mathbf{c} - \mathbf{d}) $, for all $ \mathbf{c}, \mathbf{d} \in \mathbb{F}_{q^m}^{g r} $. For a code $ \mathcal{C} \subseteq \mathbb{F}_{q^m}^{gr} $ (linear or non-linear), we define its \textit{minimum sum-rank distance} by $ {\rm d}_{SR}(\mathcal{C}) = \min \{ {\rm d}_{SR}(\mathbf{c}, \mathbf{d}) \mid \mathbf{c}, \mathbf{d} \in \mathcal{C}, \mathbf{c} \neq \mathbf{d} \} $. The number $ g $ will be called \textit{the number of matrix sets}. If the context is clear, we will not specify the partition $ (g,r) $ nor the base field $ \mathbb{F}_q $. 
\end{definition}
%

Observe that the Hamming metric \cite{hamming} and the rank metric \cite{delsartebilinear, gabidulin, roth} are recovered from the sum-rank metric by setting $ r = 1 $ and $ g = 1 $, respectively. 

We have the following extension of the Singleton bound from the Hamming metric (Proposition \ref{prop singleton bound}) to the sum-rank metric, given in \cite[Cor. 2]{universal-lrc}.

\begin{proposition} [\textbf{Singleton bound \textbf{\cite{universal-lrc}}}] \label{prop sum rank singleton}
Let $ \mathcal{C} \subseteq \mathbb{F}_{q^m}^{gr} $ be a (linear or non-linear) code. For the partition $ (g,r) $ and the base field $ \mathbb{F}_q $, we have
\begin{equation}
| \mathcal{C} | \leq q^{m (gr - {\rm d}_{SR}(\mathcal{C}) + 1)}.
\label{eq sum-rank singleton bound}
\end{equation}
Furthermore, equality holds in (\ref{eq sum-rank singleton bound}) if, and only if, $ \mathcal{C} \cdot {\rm diag}(A_1, A_2, \ldots, A_g) \subseteq \mathbb{F}_{q^m}^{gr} $ is MDS, for all $ A_1, A_2, \ldots, A_g \in {\rm GL}_{r}(\mathbb{F}_q) $.
\end{proposition}

The main objects of study in this manuscript are \textit{maximum sum-rank distance (MSRD) codes}, introduced in \cite[Th. 4]{linearizedRS}, which extend MDS codes.

\begin{definition} [\textbf{MSRD codes \cite{linearizedRS}}] \label{def msrd codes}
For a positive integer $ g $, we say that a (linear or non-linear) code $ \mathcal{C} \subseteq \mathbb{F}_{q^m}^{gr} $ is maximum sum-rank distance (MSRD), for the partition $ (g,r) $ and the base field $ \mathbb{F}_q $, if equality holds in (\ref{eq sum-rank singleton bound}).
\end{definition}
%
%

By \cite[Cor. 3]{universal-lrc}, $ m \geq r $ is required by any MSRD code of minimum sum-rank distance larger than $ 1 $. Thus we will assume $ m \geq r $ from now on.

The following result was proven in \cite[Th. 5]{gsrws}.

\begin{lemma} [\textbf{\cite{gsrws}}] \label{lemma dual of msrd is msrd}
A linear code $ \mathcal{C} \subseteq \mathbb{F}_{q^m}^{gr} $ is MSRD if, and only if, its dual $ \mathcal{C}^\perp \subseteq \mathbb{F}_{q^m}^{gr} $ is MSRD, in both cases for the partition $ (g,r) $ and base field $ \mathbb{F}_q $.
\end{lemma}

The previous lemma will be useful for our purposes, since we will construct MSRD codes by giving their parity-check matrices without worrying about computing their generator matrices, and proving that the parity-check matrices generate MSRD codes. Giving parity-check matrices will allow us to obtain higher information rates for smaller field sizes.
%

\subsection{PMDS codes} \label{subsec pmds codes}

We next revisit \textit{locally repairable codes} \cite{gopalan, kamath} and \textit{PMDS codes} \cite{blaum-RAID, gopalan-MR}. 

\begin{definition}[\textbf{Locally repairable codes \cite{gopalan, kamath}}] \label{def LRC}
Fix positive integers $ g $, $ r $ and $ \delta $, and set $ \nu = r + \delta - 1 $. A code $ \mathcal{C} \subseteq \mathbb{F}^n $ is a locally repairable code (LRC) with $ (r,\delta) $-localities if $ n = g \nu $ and we may partition $ [n] = \Gamma_1 \cup \Gamma_2 \cup \ldots \cup \Gamma_g $ such that
$$ \Gamma_i = [ (i-1)\nu + 1, i \nu ] \quad \textrm{and} \quad {\rm d}_H(\mathcal{C}_{\Gamma_i}) \geq \delta, $$
where $ \mathcal{C}_{\Gamma_i} \subseteq \mathbb{F}^\nu $ denotes the projection of $ \mathcal{C} $ onto the coordinates in $ \Gamma_i $, for $ i = 1,2, \ldots, g $. The set $ \Gamma_i $ is called the $ i $th \textit{local set} and $ \nu $ is the \textit{local-set size}. In many occasions, we only use the term \textit{locality} for the number $ r $, whereas $ \delta $ is called the \textit{local distance}. 
\end{definition}

LRCs as in \cite{gopalan, kamath} do not require pair-wise disjoint local sets, but we consider only this case since it is required for PMDS codes. Moreover, local sets need not be of the same size for PMDS codes \cite[Def. 5]{universal-lrc}, but we only consider this case for simplicity. \textit{Partial MDS (PMDS) codes}, introduced in \cite{blaum-RAID, gopalan-MR}, are those LRCs that may correct any erasure pattern that is information-theoretically correctable given the locality constraints in Definition \ref{def LRC}. Such patterns are exactly those with $ \delta - 1 $ erasures per local set and an extra $ h = gr - k $ erasures anywhere else, where $ k $ is the code dimension. This is equivalent to the following definition.

\begin{definition} [\textbf{PMDS codes \cite{blaum-RAID, gopalan-MR}}] \label{def PMDS codes}
A linear code $ \mathcal{C} \subseteq \mathbb{F}^n $ is a \textit{partial MDS (PMDS) code} with $ (r,\delta) $-localities if it is an LRC with $ (r,\delta) $-localities and, for any $ \Delta_i \subseteq \Gamma_i $ with $ | \Delta_i | = r $, for $ i = 1,2, \ldots,g $, the restricted code $ \mathcal{C}_\Delta \subseteq \mathbb{F}^{g r} $ is MDS, where $ \Delta = \bigcup_{i=1}^g \Delta_i $.
\end{definition} 

The following is Construction 1 in \cite{universal-lrc}.

\begin{construction} [\textbf{\cite{universal-lrc}}] \label{construction 1}
Fix positive integers $ g $ and $ r $, a \textit{base field} size $ q $ and an extension degree $ m \geq r $. The \textit{field of linearity} of our target codes is $ \mathbb{F} = \mathbb{F}_{q^m} $. Choose:
\begin{enumerate}
\item
\textit{Outer code}: A linear MSRD code $ \mathcal{C}_{out} \subseteq \mathbb{F}_{q^m}^{g r} $ for the partition $ (g,r) $ over $ \mathbb{F}_q $.
\item
\textit{Local code}: A linear MDS code $ \mathcal{C}_{loc} \subseteq \mathbb{F}_q^\nu $ of dimension $ r $ (over $ \mathbb{F}_q $).
\item
\textit{Global code}: Let $ \mathcal{C}_{glob} \subseteq \mathbb{F}_{q^m}^n $, where $ n = g \nu $, be given by
$$ \mathcal{C}_{glob} = \mathcal{C}_{out} \cdot {\rm diag}(A, A, \ldots, A), $$
$ g $ times, where $ A \in \mathbb{F}_q^{r \times \nu} $ is an arbitrary generator matrix of $ \mathcal{C}_{loc} $.
\end{enumerate}
\end{construction}

The following result is \cite[Th. 2]{universal-lrc}.

\begin{proposition} [\textbf{\cite{universal-lrc}}] \label{prop construction 1 is pmds}
The linear code $ \mathcal{C}_{glob} \subseteq \mathbb{F}_{q^m}^n $ from Construction \ref{construction 1} has dimension $ \dim (\mathcal{C}_{glob}) = \dim (\mathcal{C}_{out}) $ and is a PMDS code with $ (r,\delta) $-localities.
\end{proposition}
%
%

\subsection{Field sizes in applications of MSRD codes} \label{subsec other apps of SR metric}

Before constructing MSRD codes, it is crucial to know what we want in an MSRD code. The parameters of the ambient space are $ m $, $ r $ (matrix sizes), $ g $ (number of matrix sets) and $ q $ (base field size). However, the computational complexity of encoding and decoding with a linear (over $ \mathbb{F}_{q^m} $) code in $ \mathbb{F}_{q^m}^{gr} $ is governed by the size of the field of linearity $ q^m $. In the case of PMDS codes and multishot network coding \cite{multishot, secure-multishot}, the base field size $ q $ is not as important. Thus, when comparing MSRD codes for such applications, $ (m_1, q_1) $ is considered better than $ (m_2, q_2) $ if $ q_1^{m_1} < q_2^{m_2} $. However, in other applications, such as space-time coding \cite{space-time-kumar, mohannad} or crisscross error correction \cite{roth, multilayer}, we may not have such a flexibility on the pair of parameters $ (m,q) $, since $ \mathbb{F}_q $ may be fixed (it corresponds to the constellation in space-time coding and the array alphabet in crisscross error correction). Thus in such applications, if we fix $ q $, then it is desirable to obtain linear MSRD codes with smallest possible value of $ m $. This is because of the next proposition, which implies that if we find a linear MSRD code for a pair $ (q,m) $, then we may easily obtain a linear MSRD code for the pair $ (q, mM) $, for any positive integer $ M $. The proof is straightforward from the characterization in Proposition \ref{prop sum rank singleton}.

\begin{proposition}
For a linear code $ \mathcal{C} \subseteq \mathbb{F}_{q^m}^{gr} $, define $ \mathcal{C} \otimes \mathbb{F}_{q^{mM}} \subseteq \mathbb{F}_{q^{mM}}^{gr} $ as the $ \mathbb{F}_{q^{mM}} $-linear code with the same generator matrix as $ \mathcal{C} $ (which has entries in $ \mathbb{F}_{q^m} $). Then $ \dim_{\mathbb{F}_{q^m}} \left( \mathcal{C} \right) = \dim_{\mathbb{F}_{q^{mM}}} \left( \mathcal{C} \otimes \mathbb{F}_{q^{mM}} \right) $ and $ \mathcal{C} $ is MSRD if, and only if, so is $ \mathcal{C} \otimes \mathbb{F}_{q^{mM}} $, in both cases for the length partition $ (g, r) $ over the field $ \mathbb{F}_q $.
\end{proposition}

A difficult research problem, open in most cases, is to determine constraints in $ m $, $ q $ and $ q^m $ for the existence of MSRD codes and PMDS codes. This problem is a highly non-trivial extension of the \textit{MDS conjecture} (not even the asymptotic order of the size $ q^m $ of possible MSRD or PMDS codes is known in general, whereas we know that MDS codes exist if, and only if, the code length is at most linear in the field size \cite{reed-solomon, singleton}).

For PMDS codes, bounds on field sizes were given in \cite[Th. 3.5 and 3.8]{gopi}, and the case of one global parity is completely solved in \cite{complete}. For MSRD codes, bounds on the parameters were given in \cite[Th. 6.12]{alberto-fundamental}, which may be turned into bounds on field sizes. One of our MSRD codes meets the latter bounds, see Subsection \ref{subsec plugging a Hamming code}. 

\section{Extended Moore matrices} \label{sec conjugacy and ext moore matricex}

This section contains the main method for constructing parity-check matrices of MSRD codes. The section concludes with a definition of a general family of MSRD codes (Definition \ref{def general family of MSRD codes}). Such codes exist and are explicit as long as a certain sequence $ \left( \beta_1, \beta_2, \ldots, \beta_{\mu r} \right) \in \mathbb{F}_{q^m}^{\mu r} $ is known. Explicit constructions of such sequences will be deferred to Section \ref{sec technique of tensor products}.

\subsection{The definitions} \label{subsec conjugacy the definitions}

We fix the field automorphism $ \sigma : \mathbb{F}_{q^m} \longrightarrow \mathbb{F}_{q^m} $ given by $ \sigma(a) = a^q $, for $ a \in \mathbb{F}_{q^m} $.
The following definition is a particular case of \cite[Eq.~(2.5)]{lam-leroy}, but already appeared in \cite{lam}.

\begin{definition} [\textbf{\cite{lam, lam-leroy}}]
We define the equivalence relation $ \sim_\sigma $ in $ \mathbb{F}_{q^m} $ as $ a \sim_\sigma b $ if there exists $ c \in \mathbb{F}_{q^m}^* $ such that $ b = \sigma(c)c^{-1} a = c^{q-1} a $, for $ a,b \in \mathbb{F}_{q^m} $.
\end{definition}

It was shown in \cite[Cor. 1]{matroidal} that there are exactly $ q-1 $ non-zero equivalence classes in $ \mathbb{F}_{q^m} $ with respect to $ \sim_\sigma $, each of size $ (q^m-1) / (q-1) $. Furthermore, they are represented by powers of a primitive element, as observed in the paragraph after \cite[Def. 2]{universal-lrc}.

\begin{lemma} [\textbf{\cite{matroidal, universal-lrc}}] \label{lemma conjugacy classes less than q}
Let $ \gamma \in \mathbb{F}_{q^m}^* $ be a primitive element of $ \mathbb{F}_{q^m} $. Then $ \gamma^0, \gamma^1, \ldots, \gamma^{q-2} $ are pair-wise non-equivalent with respect to $ \sim_\sigma $.
\end{lemma}

Moreover, the elements in $ \mathbb{F}_q^* $ represent the $ q-1 $ equivalence classes in $ \mathbb{F}_{q^m} $ with respect to $ \sim_\sigma $ if, and only if, $ q-1 $ and $ m $ are coprime \cite[Remark 27]{SR-BCH}.
%

We next define truncated norms. Again, the following definition is a particular case of \cite[Eq.~(2.3)]{lam-leroy}, but already appeared in \cite{lam}.

\begin{definition}[\textbf{\cite{lam, lam-leroy}}] \label{def truncated norms}
For $ a \in \mathbb{F}_{q^m} $ and $ i \in \mathbb{N} $, we define $ N_i(a) = \sigma^{i-1}(a) \cdots \sigma(a)a $.
\end{definition}

We may now define extended Moore matrices. 

\begin{definition} [\textbf{Extended Moore matrices}] \label{def moore matrix extended}
Let $ \mathbf{a} = ( a_1, a_2, \ldots, a_\ell ) \in (\mathbb{F}_{q^m}^*)^\ell $ be a vector of $ \ell $ pair-wise non-equivalent elements in $ \mathbb{F}_{q^m} $ with respect to $ \sim_\sigma $. Let $ \boldsymbol\beta = ( \beta_1, \beta_2, \ldots, \beta_\eta ) \in \mathbb{F}_{q^m}^{\eta} $ be an arbitrary vector, for some positive integer $ \eta $. For $ h = 1,2, \ldots, \ell \eta $, we define the \textit{extended Moore matrix} $ M_h(\mathbf{a}, \boldsymbol\beta) \in \mathbb{F}_{q^m}^{h \times (\ell \eta)} $ by $ M_h(\mathbf{a}, \boldsymbol\beta) = $
\begin{equation*}
 \left( \begin{array}{lll|c|lll}
\beta_1 & \ldots & \beta_\eta & \ldots & \beta_1 & \ldots & \beta_\eta \\
 \beta_1^q a_1 & \ldots & \beta_\eta^q a_1 & \ldots & \beta_1^q a_\ell & \ldots & \beta_\eta^q a_\ell \\
\beta_1^{q^2} N_2(a_1) & \ldots & \beta_\eta^{q^2} N_2(a_1) & \ldots &  \beta_1^{q^2} N_2(a_\ell) & \ldots &  \beta_\eta^{q^2} N_2(a_\ell) \\
\vdots & \ddots & \vdots & \ddots & \vdots & \ddots & \vdots \\
 \beta_1^{q^{h-1}} N_{h-1}(a_1) & \ldots &  \beta_\eta^{q^{h-1}} N_{h-1}(a_1) & \ldots &  \beta_1^{q^{h-1}} N_{h-1}(a_\ell) & \ldots &  \beta_\eta^{q^{h-1}} N_{h-1}(a_\ell) \\
\end{array} \right) .
\end{equation*}
\end{definition}

Such matrices extend the well known Moore matrices \cite[Lemma 3.51]{lidl} from one to several equivalence classes of $ \sim_\sigma $. They extend the matrices in \cite[p. 604]{linearizedRS} in the sense that the $ \eta $ components of $ \boldsymbol\beta \in \mathbb{F}_{q^m}^\eta $ over $ \mathbb{F}_{q^m} $ need not be linearly independent over $ \mathbb{F}_q $. The vector $ \boldsymbol\beta \in \mathbb{F}_{q^m}^\eta $ may be different (and possibly of different lengths $ \eta $) at each of the $ \ell $ blocks. All the results in this manuscript also hold in such a generality. However, we assume equal vectors $ \boldsymbol\beta \in \mathbb{F}_{q^m}^\eta $ at different blocks for simplicity.

We would like to turn the matrix $ M_h(\mathbf{a}, \boldsymbol\beta) $ into the parity-check matrix of an MSRD code. 
%
%
%
To that end, we will use Proposition \ref{prop sum rank singleton} and a characterization of when an extended Moore matrix is the parity-check matrix of an MDS code.

\subsection{MDS extended Moore matrices} \label{subsec extended moore matrices MDS}

In this subsection, we characterize when an extended Moore matrix is the parity-check matrix of an MDS code. We need the concept of $ h $-wise independence from \cite[Def. 9]{gopalan-MR}.

\begin{definition} [\textbf{\cite{gopalan-MR}}] \label{def t-wise indep}
We say that a subset $ T \subseteq \mathbb{F}_{q^m} $ is $ h $-wise independent over $ \mathbb{F}_q $ if any subset of at most $ h $ elements of $ T $ is linearly independent over $ \mathbb{F}_q $. Analogously, for a positive integer $ \eta $, we say that a vector $ \boldsymbol\beta = (\beta_1, \beta_2, \ldots, \beta_\eta) \in \mathbb{F}_{q^m}^\eta $ is $ h $-wise independent if $ T = \{ \beta_1, \beta_2, \ldots, \beta_\eta \} $ has size $ \eta $ and is $ h $-wise independent.
\end{definition}
%

We will also need the following four auxiliary lemmas. The first one is trivial. 

\begin{lemma} \label{lemma linear combination Moore matrix}
Fix integers $ 1 \leq \eta \leq h $ and $ a \in \mathbb{F}_{q^m}^* $. Assume that there exist $ \lambda_1, \lambda_2 , \ldots, \lambda_\eta \in \mathbb{F}_q $ such that $ \lambda_1 \beta_1 + \lambda_2 \beta_2 + \cdots + \lambda_\eta \beta_\eta = 0 $, for $ \beta_1, \beta_2, \ldots, \beta_\eta \in \mathbb{F}_{q^m} $. Then
$$ \left( \begin{array}{llll}
\beta_1 & \beta_2 & \ldots & \beta_\eta \\
 \beta_1^q a & \beta_2^q a & \ldots & \beta_\eta^q a \\
\beta_1^{q^2} N_2(a) & \beta_2^{q^2} N_2(a) & \ldots & \beta_\eta^{q^2} N_2(a) \\
\vdots & \vdots & \ddots & \vdots \\
 \beta_1^{q^{h-1}} N_{h-1}(a) & \beta_2^{q^{h-1}} N_{h-1}(a) & \ldots &  \beta_\eta^{q^{h-1}} N_{h-1}(a)\\
\end{array} \right) \left( \begin{array}{c}
\lambda_{1} \\
\lambda_{2} \\
\vdots \\
\lambda_\eta
\end{array} \right) = \mathbf{0} . $$
\end{lemma}

The next lemma follows immediately from the invertibility of Moore matrices \cite[Lemma 3.51]{lidl} and Lemma \ref{lemma linear combination Moore matrix}.

\begin{lemma}  \label{lemma moore}
Fix integers $ 1 \leq \eta \leq h $ and $ a \in \mathbb{F}_{q^m}^* $. The dimension of the $ \mathbb{F}_q $-linear subspace generated by $ \beta_1, \beta_2, \ldots, \beta_\eta \in \mathbb{F}_{q^m} $ equals the rank of
$$  \left( \begin{array}{llll}
\beta_1 & \beta_2 & \ldots & \beta_\eta \\
 \beta_1^q a & \beta_2^q a & \ldots & \beta_\eta^q a \\
\beta_1^{q^2} N_2(a) & \beta_2^{q^2} N_2(a) & \ldots & \beta_\eta^{q^2} N_2(a) \\
\vdots & \vdots & \ddots & \vdots \\
 \beta_1^{q^{h-1}} N_{h-1}(a) & \beta_2^{q^{h-1}} N_{h-1}(a) & \ldots &  \beta_\eta^{q^{h-1}} N_{h-1}(a)\\
\end{array} \right) \in \mathbb{F}_{q^m}^{h \times \eta} . $$
\end{lemma}

The next lemma may be easily derived by simplifying telescopic products.

\begin{lemma} \label{lemma delenclos}
With notation as in Definition \ref{def moore matrix extended}, it holds that
$$  M_h(\mathbf{a}, \boldsymbol\beta) \cdot {\rm diag} \left( \beta_1^{-1}, \ldots, \beta_\eta^{-1} | \ldots | \beta_1^{-1}, \ldots, \beta_\eta^{-1} \right) =  $$
\begin{equation*}
 \left( \begin{array}{lll|c|lll}
1 & \ldots & 1 & \ldots & 1 & \ldots & 1 \\
 \beta_1^{q-1} a_1 & \ldots & \beta_\eta^{q-1} a_1 & \ldots & \beta_1^{q-1} a_\ell & \ldots & \beta_\eta
^{q-1} a_\ell \\
 N_2(\beta_1^{q-1} a_1) & \ldots & N_2(\beta_\eta^{q-1} a_1) & \ldots &  N_2(\beta_1^{q-1} a_\ell) & \ldots &  N_2(\beta_\eta^{q-1} a_\ell) \\
\vdots & \ddots & \vdots & \ddots & \vdots & \ddots & \vdots \\
 N_{h-1}(\beta_1^{q-1} a_1) & \ldots &  N_{h-1}(\beta_\eta^{q-1} a_1) & \ldots & N_{h-1}(\beta_1^{q-1} a_\ell) & \ldots & N_{h-1}(\beta_\eta^{q-1} a_\ell) \\
\end{array} \right)  .
\end{equation*}
\end{lemma}

The next lemma is a particular case of \cite[Th. 23 (1)]{lam}.

\begin{lemma} [\textbf{\cite{lam}}] \label{lemma lam}
Let $ a_1, a_2, \ldots, a_\ell  \in \mathbb{F}_{q^m}^* $ be $ \ell $ pair-wise non-equivalent elements in $ \mathbb{F}_{q^m} $ with respect to $ \sim_\sigma $. Take positive integers $ \eta_1, \eta_2, \ldots, \eta_\ell $ and $ h = \eta_1 + \eta_2 + \cdots + \eta_\ell $, and let $ \beta_{i,j} \in \mathbb{F}_{q^m} $, for $ j = 1,2, \ldots, \eta_i $ and for $ i = 1,2, \ldots, \ell $. Then 
\begin{equation*}
{\rm Rk} \left( \begin{array}{lll|c|lll}
1 & \ldots & 1 & \ldots & 1 & \ldots & 1 \\
 \beta_{1,1}^{q-1} a_1 & \ldots & \beta_{1,\eta_1}^{q-1} a_1 & \ldots & \beta_{\ell,1}^{q-1} a_\ell & \ldots & \beta_{\ell, \eta_\ell}^{q-1} a_\ell \\
 N_2(\beta_{1,1}^{q-1} a_1) & \ldots & N_2(\beta_{1, \eta_1}^{q-1} a_1) & \ldots &  N_2(\beta_{\ell,1}^{q-1} a_\ell) & \ldots &  N_2(\beta_{\ell, \eta_\ell}^{q-1} a_\ell) \\
\vdots & \ddots & \vdots & \ddots & \vdots & \ddots & \vdots \\
 N_{h-1}(\beta_{1,1}^{q-1} a_1) & \ldots &  N_{h-1}(\beta_{1, \eta_1}^{q-1} a_1) & \ldots & N_{h-1}(\beta_{\ell,1}^{q-1} a_\ell) & \ldots & N_{h-1}(\beta_{\ell, \eta_\ell}^{q-1} a_\ell) \\
\end{array} \right) =
\end{equation*}
$$ \sum_{i=1}^\ell {\rm Rk} \left( \begin{array}{llll}
1 & 1 & \ldots & 1 \\
 \beta_{i,1}^{q-1} a_i & \beta_{i,2}^{q-1} a_i & \ldots & \beta_{i, \eta_i}^{q-1} a_i \\
 N_2(\beta_{i,1}^{q-1} a_i) & N_2(\beta_{i,2}^{q-1} a_i) & \ldots & N_2(\beta_{i,\eta_i}^{q-1} a_i) \\
\vdots & \vdots & \ddots & \vdots \\
 N_{h-1}(\beta_{i,1}^{q-1} a_i) & N_{h-1}(\beta_{i,2}^{q-1} a_i) & \ldots &  N_{h-1}(\beta_{i,\eta_i}^{q-1} a_i)
\end{array} \right). $$
\end{lemma}

The main result of this subsection is the following theorem. 

\begin{theorem} \label{th extended moore matrix is mds}
Let $ \mathbf{a} = ( a_1, a_2, \ldots, a_\ell)  \in (\mathbb{F}_{q^m}^*)^\ell $ be a vector of $ \ell $ pair-wise non-equivalent elements in $ \mathbb{F}_{q^m} $ with respect to $ \sim_\sigma $. For $ h = 1,2, \ldots, \ell \eta $, the extended Moore matrix $ M_h(\mathbf{a}, \boldsymbol\beta) \in \mathbb{F}_{q^m}^{h \times (\ell \eta)} $ from Definition \ref{def moore matrix extended} is a full-rank parity-check matrix of an MDS code if, and only if, $ \boldsymbol\beta = ( \beta_1, \beta_2, $ $ \ldots, $ $ \beta_\eta ) \in \mathbb{F}_{q^m}^\eta $ is $ h $-wise independent over $ \mathbb{F}_q $.
\end{theorem}
\begin{proof}
First, assume that $ ( \beta_1, \beta_2, \ldots, \beta_\eta ) $ is not $ h $-wise independent over $ \mathbb{F}_q $. Then $ M_h(\mathbf{a}, \boldsymbol\beta) $ contains an $ h \times h $ submatrix that is not invertible by Lemma \ref{lemma linear combination Moore matrix}.

Conversely, assume that $ ( \beta_1, \beta_2, \ldots, \beta_\eta ) $ is $ h $-wise independent over $ \mathbb{F}_q $. Take an arbitrary $ h \times h $ submatrix $ M^\prime \in \mathbb{F}_{q^m}^{h \times h} $ of $ M_h(\mathbf{a}, \boldsymbol\beta) $, and let $ 0 \leq \eta_i \leq \min \{ h, \eta \} $ be the number of columns from the $ i $th block of $ \eta $ columns in $ M_h(\mathbf{a}, \boldsymbol\beta) $ appearing in $ M^\prime $, for $ i = 1,2, \ldots, \ell $. Note that $ h = \eta_1 + \eta_2 + \cdots + \eta_\ell $. Since $ ( \beta_1, \beta_2, \ldots, \beta_\eta ) $ is $ h $-wise independent over $ \mathbb{F}_q $ and $ \eta_i \leq h $, then the $ i $th block of $ \eta_i $ columns in $ M^\prime $ forms an $  \eta_i  \times h $ matrix of full rank $ \eta_i $ by Lemma \ref{lemma moore}, for $ i = 1,2, \ldots, \ell $. Finally, by combining Lemmas \ref{lemma delenclos} and \ref{lemma lam}, we conclude that $ {\rm Rk}(M^\prime) =  \eta_1 + \eta_2 + \cdots + \eta_\ell  = h $, and therefore $ M^\prime \in \mathbb{F}_{q^m}^{h \times h} $ is invertible. Hence $ M_h(\mathbf{a}, \boldsymbol\beta) $ is MDS and we are done.
\end{proof}

\subsection{MSRD extended Moore matrices} \label{subsec extended moore matrices MSRD}

In this subsection, we characterize when an extended Moore matrix is the parity-check matrix of an MSRD code. We start by combining Proposition \ref{prop sum rank singleton}, Lemma \ref{lemma dual of msrd is msrd}, Theorem \ref{th extended moore matrix is mds} and the $ \mathbb{F}_q $-linearity of the map $ \sigma $.

\begin{proposition} \label{prop extended moore matrix is msrd}
Let $ \mathbf{a} = ( a_1, a_2, \ldots, a_\ell)  \in (\mathbb{F}_{q^m}^*)^\ell $ be a vector of $ \ell $ pair-wise non-equivalent elements in $ \mathbb{F}_{q^m} $ with respect to $ \sim_\sigma $. Let $ \boldsymbol\beta = (\beta_1, \beta_2, \ldots, \beta_{\mu r}) \in \mathbb{F}_{q^m}^{\mu r} $, for positive integers $ \mu $ and $ r $, and set $ g = \ell \mu $. For $ h = 1,2, \ldots, gr $, the extended Moore matrix $ M_h(\mathbf{a}, \boldsymbol\beta) \in \mathbb{F}_{q^m}^{h \times (gr)} $ from Definition \ref{def moore matrix extended} is a full-rank parity-check matrix of an MSRD code  for the partition $ (g,r) $ over $ \mathbb{F}_q $ if, and only if, for all $ A_1, A_2, \ldots, A_\mu \in {\rm GL}_r (\mathbb{F}_q) $, the vector
$$ (\beta_1, \beta_2, \ldots, \beta_{\mu r}) \cdot {\rm diag} (A_1, A_2, \ldots, A_\mu) \in \mathbb{F}_{q^m}^{\mu r} $$
is $ h $-wise independent over $ \mathbb{F}_q $.
\end{proposition}

Our main characterization is the following theorem.

\begin{theorem} \label{th extended moore matrix is msrd sufficient}
Let $ \mathbf{a} = ( a_1, a_2, \ldots, a_\ell)  \in (\mathbb{F}_{q^m}^*)^\ell $ be a vector of $ \ell $ pair-wise non-equivalent elements in $ \mathbb{F}_{q^m} $ with respect to $ \sim_\sigma $. Let $ \boldsymbol\beta = (\beta_1, \beta_2, \ldots, \beta_{\mu r}) \in \mathbb{F}_{q^m}^{\mu r} $, for positive integers $ \mu $ and $ r $, and set $ g = \ell \mu $. Define the $ \mathbb{F}_q $-linear subspace
\begin{equation}
 \mathcal{H}_i = \left\langle \beta_{(i-1)r+1}, \beta_{(i-1)r+2}, \ldots, \beta_{ir}  \right\rangle_{\mathbb{F}_q} \subseteq \mathbb{F}_{q^m},
\label{eq subspace from h-lin indep}
\end{equation}
for $ i = 1,2, \ldots, \mu $. Given $ 1 \leq h \leq gr $, the extended Moore matrix $ M_h(\mathbf{a}, \boldsymbol\beta) \in \mathbb{F}_{q^m}^{h \times (gr)} $ from Definition \ref{def moore matrix extended} is a full-rank parity-check matrix of an MSRD code for the partition $ (g,r) $ over $ \mathbb{F}_q $ if, and only if, the following two conditions hold for all $ i = 1,2, \ldots, \mu $:
\begin{enumerate}
\item
$ \dim_{\mathbb{F}_q}(\mathcal{H}_i) = r $, i.e., $ \beta_{(i-1)r+1}, \beta_{(i-1)r+2}, \ldots, \beta_{ir} $ are $ \mathbb{F}_q $-linearly independent, and
\item
$ \mathcal{H}_i \cap \left( \sum_{j \in \Gamma} \mathcal{H}_j \right) = \{ 0 \} $, for any set $ \Gamma \subseteq [\mu] $, such that $ i \notin \Gamma $ and $ |\Gamma| \leq \min \{ h,\mu \} -1 $.
\end{enumerate}
\end{theorem}
\begin{proof}
We prove both implications separately.

$ \Longleftarrow $): Take matrices $ A_1, A_2, \ldots, A_\mu \in {\rm GL}_r(\mathbb{F}_q) $. Condition 1 implies that $ \beta^\prime_{(i-1)r+1}, $ $ \beta^\prime_{(i-1)r+2}, $ $ \ldots, $ $ \beta^\prime_{ir} $ $ \in \mathbb{F}_{q^m} $ are linearly independent over $ \mathbb{F}_q $, where
$$ (\beta^\prime_{(i-1)r+1}, \beta^\prime_{(i-1)r+2}, \ldots, \beta^\prime_{ir}) = (\beta_{(i-1)r+1}, \beta_{(i-1)r+2}, \ldots, \beta_{ir}) \cdot A_i \in \mathbb{F}_{q^m}^r, $$
for all $ i = 1,2, \ldots, \mu $. Next, fix an index $ i = 1,2, \ldots, \mu $, and take a subset $ \Gamma \subseteq [\mu] $, such that $ i \notin \Gamma $ and $ |\Gamma| \leq \min \{ h,\mu \} -1 $. Condition 2 and the $ \mathbb{F}_q $-linear independence of each set $ \{ \beta^\prime_{(j-1)r+1}, \beta^\prime_{(j-1)r+2}, $ $ \ldots, $ $ \beta^\prime_{jr} \} $ imply that the set
\begin{equation}
\bigcup_{j \in \Gamma \cup \{ i\}} \left\lbrace \beta^\prime_{(j-1)r+1}, \beta^\prime_{(j-1)r+2}, \ldots, \beta^\prime_{jr} \right\rbrace \subseteq \mathbb{F}_{q^m}
\label{eq set to be linearly indep for Moore msrd}
\end{equation}
is linearly independent over $ \mathbb{F}_q $. Since every subset of size at most $ h $ of $ \{ \beta^\prime_1, \beta^\prime_2, \ldots, \beta^\prime_{\mu r} \} $ is contained in a set of the form (\ref{eq set to be linearly indep for Moore msrd}), we deduce that the vector $ ( \beta^\prime_1, \beta^\prime_2, \ldots, \beta^\prime_{\mu r} ) $ is $ h $-wise linearly independent over $ \mathbb{F}_q $. Hence the extended Moore matrix $ M_h(\mathbf{a}, \boldsymbol\beta) \in \mathbb{F}_{q^m}^{h \times N} $ is MSRD by Proposition \ref{prop extended moore matrix is msrd}.

$ \Longrightarrow $): Assume first that Condition 1 does not hold for some $ i = 1,2, \ldots, \mu $. Without loss of generality, we may assume that there exist $ \lambda_1, \lambda_2, \ldots, \lambda_{r-1} \in \mathbb{F}_q $ such that
$$ \sum_{j=1}^{r-1} \lambda_j \beta_{(i-1)r + j} + \beta_{ir} = 0. $$
Thus if we define the invertible matrix
$$ A_i = \left( \begin{array}{c|c}
I_{r-1} & \begin{array}{c}
\lambda_{1} \\
\vdots \\
\lambda_{r-1}
\end{array} \\
\hline
\begin{array}{ccc}
0 & \ldots & 0
\end{array} & 1
\end{array} \right) \in {\rm GL}_r(\mathbb{F}_q), $$
where $ I_{r-1} \in {\rm GL}_{r-1}(\mathbb{F}_q) $ denotes the $ (r-1) \times (r-1) $ identity matrix, then it holds that
$$ (\beta_{(i-1)r+1}, \ldots, \beta_{ir -1}, \beta_{ir}) \cdot A_i = (\beta_{(i-1)r+1}, \ldots, \beta_{ir -1}, 0). $$
Clearly, $ (\beta_{(i-1)r+1}, \ldots, \beta_{ir -1}, 0) \in \mathbb{F}_{q^m}^r $ is not $ h $-wise independent, thus $ M_h(\mathbf{a}, \boldsymbol\beta) $ is not MSRD by Proposition \ref{prop extended moore matrix is msrd}.

Next, assume that Condition 2 does not hold for some $ i = 1,2, \ldots, \mu $. Then we may assume, without loss of generality, that there exists a subset $ \Gamma \subseteq [\mu] $ such that $ i \in \Gamma $, $ |\Gamma| \leq h $, and there exist $ \lambda_{j,u} \in \mathbb{F}_q $, for $ u = 1,2, \ldots, r $, for $ j \in \Gamma $, such that $ \lambda_{j,r} = 1 $, for $ j \in \Gamma $, and
$$ \sum_{j \in \Gamma} \sum_{u=1}^r \lambda_{j,u} \beta_{(j-1)r + u} = 0. $$
Define, for each $ j \in \Gamma $, the invertible matrix
$$ A_j = \left( \begin{array}{c|c}
I_{r-1} & \begin{array}{c}
\lambda_{j,1} \\
\vdots \\
\lambda_{j,r-1}
\end{array} \\
\hline
\begin{array}{ccc}
0 & \ldots & 0
\end{array} & 1
\end{array} \right) \in {\rm GL}_r(\mathbb{F}_q), $$
and define, for convenience, $ A_j = I_r \in {\rm GL}_r(\mathbb{F}_q) $ if $ j \notin \Gamma $. If we set
$$ (\beta^\prime_1, \beta^\prime_2, \ldots, \beta^\prime_{\mu r}) = (\beta_1, \beta_2, \ldots, \beta_{\mu r}) \cdot {\rm diag}(A_1, A_2, \ldots, A_\mu), $$
then it holds that
$$ \sum_{j \in \Gamma} \beta^\prime_{jr} = \sum_{j \in \Gamma} \sum_{u=1}^r \lambda_{j,u} \beta_{(j-1)r + u} = 0. $$
Since $ |\Gamma| \leq h $, then the vector $ (\beta^\prime_1, \beta^\prime_2, \ldots, \beta^\prime_{\mu r}) $ is not $ h $-wise independent over $ \mathbb{F}_q $, hence $ M_h(\mathbf{a}, \boldsymbol\beta) $ is not MSRD by Proposition \ref{prop extended moore matrix is msrd}.
\end{proof}

Therefore, we have the following general family of MSRD codes. Since we may puncture a linear MSRD code to obtain a shorter linear MSRD code \cite[Cor. 7]{gsrws}, we will assume from now on that $ \ell = q-1 $, which is the maximum value of $ \ell $ that we may choose as explained in Subsection \ref{subsec conjugacy the definitions}.

\begin{definition} \label{def general family of MSRD codes}
Let $ \mathbf{a} = ( a_1, a_2, \ldots, a_{q-1} ) \in (\mathbb{F}_{q^m}^*)^{q-1} $ be a vector of $ q-1 $ pair-wise non-equivalent elements in $ \mathbb{F}_{q^m} $ with respect to $ \sim_\sigma $. Let $ \boldsymbol\beta = (\beta_1, \beta_2, \ldots, \beta_{\mu r}) \in \mathbb{F}_{q^m}^{\mu r} $ satisfy Conditions 1 and 2 in Theorem \ref{th extended moore matrix is msrd sufficient}, and set $ g = (q-1) \mu $. For $ h = 1,2, \ldots, gr $, we define the following $ ( gr -h) $-dimensional linear MSRD code for the partition $ (g,r) $ over the base field $ \mathbb{F}_q $:
$$ \mathcal{C}_k(\mathbf{a}, \boldsymbol\beta) = \left\lbrace \mathbf{y} \in \mathbb{F}_{q^m}^{gr} \mid M_h(\mathbf{a}, \boldsymbol\beta) \mathbf{y}^T = \mathbf{0} \right\rbrace . $$
\end{definition}
%
%
%

\section{Explicit constructions of MSRD codes} \label{sec technique of tensor products} 

What is missing in Definition \ref{def general family of MSRD codes} is finding the sequence $ (\beta_1, \beta_2, \ldots, \beta_{\mu r}) \in \mathbb{F}_{q^m}^{\mu r} $ satisfying Conditions 1 and 2 in Theorem \ref{th extended moore matrix is msrd sufficient}. In this section, we provide a technique for \textit{explicitly} constructing such sequences. This method provides several \textit{explicit} subfamilies of the codes in Definition \ref{def general family of MSRD codes}, where the vector $ \mathbf{a} \in (\mathbb{F}_{q^m}^*)^{q-1} $ can be explicitly chosen as in Lemma \ref{lemma conjugacy classes less than q} or using all of the elements of $ \mathbb{F}_q^* $ when $ q-1 $ and $ m $ are coprime \cite[Remark 27]{SR-BCH}.

\subsection{Tensor products or field reduction} \label{subsec technique of tensor products} 

In this subsection, we explore tensor products of sequences over $ \mathbb{F}_{q^r} $ and $ \mathbb{F}_{q^m} $. This technique is inspired by \cite[Sec. IV-B]{gabrys}. However, the codes obtained in \cite[Sec. IV-B]{gabrys} and in this work are not equivalent (by inspecting their parameters). Our technique is also connected to finite geometry. Families of subspaces satisfying Conditions 1 and 2 as in Theorem \ref{th extended moore matrix is msrd sufficient} are called \textit{arcs} or \textit{pseudo-arcs} in finite geometry \cite[Sec. 4]{ball-additive}. The technique that we use here to construct them is equivalent to \textit{field reduction} \cite[Sec. 5]{ball-additive}. 

For the remainder of this section, we will fix an ordered basis $ \boldsymbol\alpha = ( \alpha_1, \alpha_2, \ldots, \alpha_r ) \in \mathbb{F}_{q^r}^r $ of $ \mathbb{F}_{q^r} $ over $ \mathbb{F}_q $. We will also assume from now on that $ m = r \rho $ (hence $ \mathbb{F}_{q^r} \subseteq \mathbb{F}_{q^m} $), for some positive integer $ \rho $. Choose a vector
\begin{equation}
\boldsymbol\gamma = (\gamma_1, \gamma_2, \ldots, \gamma_\mu) \in \mathbb{F}_{q^m}^\mu .
\label{eq def vector gamma for tensor prod}
\end{equation}
Define the \textit{tensor product} of $ \boldsymbol\alpha $ with $ \boldsymbol\gamma $ as
\begin{equation}
(\beta_1, \beta_2, \ldots, \beta_{\mu r}) = \boldsymbol\alpha \otimes \boldsymbol\gamma = (\alpha_1 \gamma_1, \ldots, \alpha_r \gamma_1 | \ldots | \alpha_1 \gamma_\mu, \ldots, \alpha_r \gamma_\mu) \in \mathbb{F}_{q^m}^{\mu r}.
\label{eq def tensor prod}
\end{equation}
In other words, $ (\beta_{(i-1)r+1}, \beta_{(i-1)r+2}, \ldots, \beta_{ir}) = \gamma_i \boldsymbol\alpha \in \mathbb{F}_{q^m}^r $, for $ i = 1,2, \ldots, \mu $.

The proof of the following theorem is straightforward and is left to the reader.

\begin{theorem} \label{th tensor prod for msrd}
The vector $ (\beta_1, \beta_2, \ldots, \beta_{\mu r}) \in \mathbb{F}_{q^m}^{\mu r} $ in (\ref{eq def tensor prod}) satisfies Conditions 1 and 2 in Theorem \ref{th extended moore matrix is msrd sufficient}, for $ i = 1,2, \ldots, \mu $, if and only if, the vector $ \boldsymbol\gamma = (\gamma_1, \gamma_2, \ldots, \gamma_\mu) \in \mathbb{F}_{q^m}^\mu $ is $ t $-wise independent over $ \mathbb{F}_{q^r} $, for $ t = \min \{ h,\mu \} $.
\end{theorem}
%

We may construct $ t $-wise independent vectors by using the following equivalence, which  has been used previously in the PMDS literature in \cite[Th. 17]{gopalan-MR}, \cite[Lemma 7]{gabrys} and \cite{guruswami-PMDS}. 
%
We will assume from now on that $ \rho \leq \mu $. The following result follows by combining Definition \ref{def t-wise indep}, the $ \mathbb{F}_{q^r} $-linearity of the map $ M_{\boldsymbol\delta} $ and \cite[Th. 10, p. 33]{macwilliamsbook}.

\begin{lemma} \label{lemma t-wise indep means minimum hamming dist}
Let $ \boldsymbol\delta \in \mathbb{F}_{q^m}^\rho $ be an ordered basis of $ \mathbb{F}_{q^m} $ over $ \mathbb{F}_{q^r} $. Consider the matrix representation map $ M_{\boldsymbol\delta} : \mathbb{F}_{q^m}^\mu \longrightarrow \mathbb{F}_{q^r}^{\rho \times \mu} $, as in (\ref{eq def matrix representation map}), and define $ H_{\boldsymbol\gamma} = M_{\boldsymbol\delta}(\boldsymbol\gamma) \in \mathbb{F}_{q^r}^{\rho \times \mu} $. The vector $ \boldsymbol\gamma \in \mathbb{F}_{q^m}^\mu $ is $ t $-wise independent over $ \mathbb{F}_{q^r} $ if, and only if, $ {\rm d}_H(\mathcal{C}_{\boldsymbol\gamma}) \geq t+1 $, for 
\begin{equation}
\mathcal{C}_{\boldsymbol\gamma} = \left\lbrace \mathbf{y} \in \mathbb{F}_{q^r}^\mu \mid H_{\boldsymbol\gamma} \mathbf{y}^T = \mathbf{0} \right\rbrace \subseteq \mathbb{F}_{q^r}^\mu .
\label{eq vector gamma as code}
\end{equation}
\end{lemma}

In conclusion, to construct $ \boldsymbol\gamma \in \mathbb{F}_{q^r}^\mu $, we may choose any $ \mathbb{F}_{q^r} $-linear code $ \mathcal{C}_{\boldsymbol\gamma} \subseteq \mathbb{F}_{q^r}^\mu $ of dimension $ \mu - \rho $ and minimum Hamming distance at least $ t + 1 $. As in related works \cite{gabrys, gopalan-MR, guruswami-PMDS}, the field size $ q^m = \left( q^r \right)^\rho $ has as exponent $ \rho $ the co-dimension of the code $ \mathcal{C}_{\boldsymbol\gamma} $. Therefore such a co-dimension needs to be as small as possible. 

\subsection{Using trivial codes: Recovering linearized RS codes} \label{subsec plugging a basis}

As a first choice of $ \mathcal{C}_{\boldsymbol\gamma} $, we choose a trivial code $ \mathcal{C}_{\boldsymbol\gamma} = \{ \mathbf{0} \} $ and recover duals of linearized Reed-Solomon codes \cite{linearizedRS}. This is the only case where the MSRD codes that we obtain are not new. This section is included as a remark. 

\begin{theorem} \label{th using trivial codes}
Choose $ \mu = \rho = 1 $. Thus $ m = r $, $ \mathcal{C}_{\boldsymbol\gamma} = \{ \mathbf{0} \} \subseteq \mathbb{F}_{q^r}^1 $ and
$$ (\beta_1, \beta_2, \ldots, \beta_r) = (\alpha_1, \alpha_2, \ldots, \alpha_r) \in \mathbb{F}_{q^r}^r. $$
In particular, $ g = q-1 $. Then the MSRD code $ \mathcal{C}_k(\mathbf{a}, \boldsymbol\beta) \subseteq \mathbb{F}_{q^r}^{g r} $ is the dual of a \textit{linearized Reed-Solomon code} \cite[Def. 31]{linearizedRS}. The co-dimension $ h $ is arbitrary with $ 1 \leq h \leq gr-1 $, and the field of linearity has size
\begin{equation}
|\mathbb{F}_{q^m}| = q^r .
\label{eq field size 0}
\end{equation}
\end{theorem} 
%

By \cite[Th. 4]{secure-multishot}, such duals are precisely linearized Reed-Solomon codes for $ \mathbf{a} \in (\mathbb{F}_{q^r}^*)^{q-1} $ as in Lemma \ref{lemma conjugacy classes less than q} (see also \cite[Prop. 38]{SR-BCH} for other cases and \cite{caruso-dual} in general). 

Notice that the PMDS codes from \cite{universal-lrc} are precisely those obtained from Construction \ref{construction 1} when taking the MSRD codes from Theorem \ref{th using trivial codes} as outer codes.

\subsection{Using MDS codes} \label{subsec plugging an MDS code}

In this subsection, we explore the case where $ \mathcal{C}_{\boldsymbol\gamma} $ is an MDS code.

\begin{theorem} \label{th using mds codes}
Choose $ \mu = q^r + 1 $, $ g = (q-1) \left( q^r + 1 \right) $ and $ \rho = t = \min \{ h, \mu \} $, being $ h $ arbitrary with $ 1 \leq h \leq gr-1 $. Choose $ \mathcal{C}_{\boldsymbol\gamma} \subseteq \mathbb{F}_{q^r}^\mu $ in (\ref{eq vector gamma as code}) as an MDS code of dimension $ \mu - t $, thus $ {\rm d}_H(\mathcal{C}_{\boldsymbol\gamma}) = t+1 $. For instance, $ \mathcal{C}_{\boldsymbol\gamma} $ can be chosen as the projective extension \cite[Th. 5.3.4]{pless} of a classical Reed-Solomon code \cite{reed-solomon}. Then the field of linearity of the MSRD code $ \mathcal{C}_k(\mathbf{a}, \boldsymbol\beta) \subseteq \mathbb{F}_{q^m}^{g r} $ has size 
\begin{equation}
|\mathbb{F}_{q^m}| = \left( \frac{g}{q-1} - 1 \right) ^{\min \left\lbrace h, \frac{g}{q-1} \right\rbrace } .
\label{eq field size 2 . 75}
\end{equation}
\end{theorem}
%

We now plug the MSRD codes from Theorem \ref{th using mds codes} into Construction \ref{construction 1}. The following corollary holds by Proposition \ref{prop construction 1 is pmds}. 

\begin{corollary}
In Construction \ref{construction 1}, choose $ \mathcal{C}_{out} = \mathcal{C}_k(\mathbf{a}, \boldsymbol\beta) \subseteq \mathbb{F}_{q^m}^{g r} $ to be the MSRD code in Theorem \ref{th using mds codes}, and let $ q $ be a power of $ 2 $ satisfying $ q > \nu $. Then $ \mathcal{C}_{glob} \subseteq \mathbb{F}_{q^m}^{g \nu} $ in Construction \ref{construction 1} is a PMDS code with $ (r,\delta) $-localities, and its field of linearity has size
\begin{equation}
|\mathbb{F}_{q^m}| = q^{r \min \{ h, \mu\}} \leq \max \left\lbrace \left( 2 \nu \right)^r , \left\lfloor \frac{g}{\nu} \right\rfloor - 1 \right\rbrace ^{\min \left\lbrace h,  \left\lfloor \frac{g}{\nu} \right\rfloor \right\rbrace } .
\label{eq field size 3}
\end{equation}
\end{corollary}

\subsection{Using Hamming codes} \label{subsec plugging a Hamming code}

We now investigate the case $ h = 2 $. As we show next, we obtain MSRD codes with arbitrary parameters except for $ h = 2 $ and with field sizes $ q^m $ that are linear in $ g $. In addition, such MSRD codes meet the bounds \cite[Th. 6.12]{alberto-fundamental}.
%
%

\begin{theorem} \label{th using hamming codes}
Consider $ h = 2 $ and $ 1 \leq \rho < \mu $, and choose $ \mathcal{C}_{\boldsymbol\gamma} \subseteq \mathbb{F}_{q^r}^\mu $ in (\ref{eq vector gamma as code}) as a $ (\mu - \rho) $-dimensional Hamming code. In other words, choose the vector $ \boldsymbol\gamma = (\gamma_1, \gamma_2, \ldots, \gamma_\mu) \in \left( \mathbb{F}_{q^{r \rho}}^* \right) ^\mu $ in (\ref{eq def vector gamma for tensor prod}) such that its components form the projective space $ \mathbb{P}_{\mathbb{F}_{q^r}}(\mathbb{F}_{q^{r \rho}}) = \{ [\gamma_1], [\gamma_2], $ $ \ldots, [\gamma_{\mu}] \} $, where $ [\gamma] = \{ \lambda \gamma \in \mathbb{F}_{q^{r \rho}}^* \mid \lambda \in \mathbb{F}_{q^r}^* \} $, for $ \gamma \in \mathbb{F}_{q^{r \rho}}^* $. 
The MSRD code $ \mathcal{C}_k(\mathbf{a}, \boldsymbol\beta) \subseteq \mathbb{F}_{q^m}^{g r} $ satisfies that $ {\rm d}_{SR}(\mathcal{C}_k(\mathbf{a}, \boldsymbol\beta)) = 3 $ and 
its field of linearity has size 
\begin{equation}
|\mathbb{F}_{q^m}| = \frac{q^r-1}{q - 1} \cdot g + 1 .
\label{eq field size 5 msrd}
\end{equation}
\end{theorem}
%
%

For $ r \geq 2 $, the linear MSRD codes from Theorem \ref{th using hamming codes} meet the following bound on field sizes, which follows from \cite[Th. 6.12]{alberto-fundamental}.

\begin{proposition} [\textbf{\cite{alberto-fundamental}}] \label{prop bounds alberto}
For positive integers $ m $, $ r $ and $ g $, let $ \mathcal{C} \subseteq \mathbb{F}_{q^m}^{gr} $ be a (linear or non-linear) MSRD code. If $ {\rm d}_{SR}(\mathcal{C}) = 3 $, $ r $ divides $ m $ and $ r \geq 2 $, then
\begin{equation}
g \leq (q-1) \cdot \frac{q^m - 1}{q^r - 1} , \quad \textrm{or} \quad q^m \geq \frac{q^r - 1}{q - 1} \cdot g + 1 .
\label{eq bound 2 alberto refined}
\end{equation}
\end{proposition}

We will not provide the corresponding construction of PMDS codes via Construction \ref{construction 1}, as there exist linear PMDS codes for $ h = 2 $ with smaller field sizes \cite{blaum-twoparities, neri-PMDS}.

\subsection{Using BCH codes} \label{subsec plugging a long BCH code}

In this subsection, we explore the case where $ \mathcal{C}_{\boldsymbol\gamma} \subseteq \mathbb{F}_{q^r}^\mu $ is a BCH code. We will set in this subsection $ \mu = q^{rs} - 1 $, for an arbitrary positive integer $ s $. 
Consider the code $ \mathcal{C}_{\boldsymbol\gamma} \subseteq \mathbb{F}_{q^r}^\mu $ in (\ref{eq vector gamma as code}) to be a \textit{BCH code}, see \cite[Sec. 7.6]{macwilliamsbook} \cite[Sec. 4.5 \& Ch. 5]{pless}. By the BCH bound \cite[Sec. 7.6, Th. 8]{macwilliamsbook} \cite[Th. 4.5.3]{pless}, we have that $ {\rm d}_H \left( \mathcal{C}_{\boldsymbol\gamma} \right) \geq \partial $ if the minimal generator polynomial of $ \mathcal{C}_{\boldsymbol\gamma} $ vanishes in $ 1, a, a^{2}, \ldots, a^{\partial - 2} \in \mathbb{F}_{q^{rs}} $, for an integer $ 2 \leq \partial \leq \mu +1 $, where $ a \in \mathbb{F}_{q^{rs}}^* $ is a primitive element. If we choose $ \mathcal{C}_{\boldsymbol\gamma} $ to be the largest BCH code whose minimal generator polynomial has such roots, then by \cite[Th. 4.2.1]{pless}, we have that
$$ \rho = \mu - \dim \left( \mathcal{C}_{\boldsymbol\gamma} \right) = | C_0 \cup C_{1} \cup C_{2} \cup \ldots \cup C_{\partial -2} |, $$
where $ C_i = \{ i, i q^{r}, i q^{2r}, i q^{3r}, \ldots \} $ modulo $ \mu $, see \cite[Sec. 7.5, p. 197]{macwilliamsbook} \cite[Sec. 4.1]{pless}, 
for $ i = 0,1, \ldots, \mu - 1 $. The integer $ \partial $ is called the \textit{prescribed distance} of the BCH code $ \mathcal{C}_{\boldsymbol\gamma} $, and the set $ C_0 \cup C_{1} \cup C_{2} \cup \ldots \cup C_{\partial -2} $ is called the \textit{defining set} of $ \mathcal{C}_{\boldsymbol\gamma} $.
%
%

Set $ \partial = t+1 = \min \{h,\mu \} + 1 $. To upper bound $ \rho $, we only need to upper bound the size of the defining set $ C_0 \cup C_{1} \cup C_{2} \cup \ldots \cup C_{t-1} $. It is trivial and well known that $ C_0 = \{ 0 \} $, $ |C_i| \leq s $ and $ C_{i q^r} = C_i $, for $ i = 0,1,2, \ldots, \mu - 1 $. Therefore, $ |C_0 | = 1 $ and we may remove from $ C_0 \cup C_1 \cup C_2 \cup \ldots \cup C_{t-1} $ each cyclotomic coset $ C_i $ where $ i $ is a multiple of $ q^r $. Hence
$$ | C_0 \cup C_1 \cup C_2 \cup \ldots \cup C_{t-1} | \leq 1 + s \cdot \left\lceil \frac{q^r - 1}{q^r} \cdot (t-1) \right\rceil . $$

Therefore, we have proven the following theorem.

\begin{theorem} \label{th using bch codes better}
Let $ s \in \mathbb{N}_+ $, $ \mu = q^{rs} - 1 $, $ g = (q-1) (q^{rs} - 1) $ and let $ h $ be arbitrary with $ 1 \leq h \leq gr - 1 $. Choose $ \mathcal{C}_{\boldsymbol\gamma} \subseteq \mathbb{F}_{q^r}^\mu $ in (\ref{eq vector gamma as code}) to be a BCH code, as above, with defining set $ C_0 \cup C_{1} \cup C_{2} \cup \ldots \cup C_{t-1} $. Then the field of linearity of the MSRD code $ \mathcal{C}_k(\mathbf{a}, \boldsymbol\beta) $ has size
\begin{equation}
|\mathbb{F}_{q^m}| = q^{r \rho} \leq q^r \cdot \left( \frac{g}{q-1} + 1 \right) ^{ \left\lceil \frac{q^r - 1}{q^r} (h -1) \right\rceil} .
\label{eq field size 8}
\end{equation}
\end{theorem}

We now plug the MSRD codes from Theorem \ref{th using bch codes better} into Construction \ref{construction 1}. The following corollary holds by Proposition \ref{prop construction 1 is pmds}. 

\begin{corollary}
In Construction \ref{construction 1}, choose $ \mathcal{C}_{out} = \mathcal{C}_k(\mathbf{a}, \boldsymbol\beta) \subseteq \mathbb{F}_{q^m}^{g r} $ to be the MSRD code in Theorem \ref{th using bch codes better}. We further assume that $ q $ is the smallest power of $ 2 $ such that $ q > \nu = r + \delta - 1 $. Then $ \mathcal{C}_{glob} \subseteq \mathbb{F}_{q^m}^{g \nu} $ in Construction \ref{construction 1} is a PMDS code with $ (r,\delta) $-localities, and its field of linearity has size
\begin{equation}
|\mathbb{F}_{q^m}| \leq \left( 2 \nu \right)^r \cdot \left( \left\lfloor \frac{g}{\nu} \right\rfloor + 1 \right) ^{ \left\lceil \frac{q^r - 1}{q^r} (h -1) \right\rceil} .
\label{eq field size 9}
\end{equation}
\end{corollary}


\subsection{Using Algebraic-Geometry (AG) codes} \label{subsec plugging an Algebraic-Geometry code}

In this subsection, we explore the case where $ \mathcal{C}_{\boldsymbol\gamma} \subseteq \mathbb{F}_{q^r}^\mu $ is an Algebraic-Geometry code, or AG code for short. AG codes have been proposed to construct PMDS codes in \cite{neri-PMDS, guruswami-PMDS}, but they have not yet been used to construct MSRD codes. 

We will only need to describe the parameters of the considered AG codes. 
%
For further details, the reader is referred to \cite{stichtenothbook}. Consider an \textit{irreducible projective curve} $ \mathcal{X} $ over $ \mathbb{F}_{q^r} $ with \textit{algebraic function field} $ \mathcal{F} $, and let $ \mathfrak{g} = \mathfrak{g}(\mathcal{X}) = \mathfrak{g}(\mathcal{F}) $ be its \textit{genus}. Points in the curve $ \mathcal{X} $ with coordinates over $ \mathbb{F}_{q^r} $ will be called \textit{rational}.
We will mainly use the following lemma on the parameters of AG codes, which follows from \cite[Cor. 2.2.3]{stichtenothbook}.

\begin{lemma} [\textbf{\cite{stichtenothbook}}] \label{lemma goppa}
Let $ \mu \geq 2 \mathfrak{g} $ such that the number of rational points of $ \mathcal{X} $ is at least $ \mu + 1 $. For any $ \mathfrak{g} \leq k \leq \mu - \mathfrak{g} $, we may construct a $ k $-dimensional linear AG code $ \mathcal{C} \subseteq \mathbb{F}_{q^r}^\mu $ from the curve $ \mathcal{X} $ such that $ {\rm d}_H(\mathcal{C}) \geq \mu - k - \mathfrak{g} + 1 $.
\end{lemma} 
%
%
%
From this lemma, we may deduce the following theorem.

\begin{theorem} \label{th using ag codes}
Assume that $ \mu - h \geq 2 \mathfrak{g} $, in particular $ t = \min \{ h,\mu \} = h $, and assume that the number of rational points of $ \mathcal{X} $ is at least $ \mu + 1 $. Define $ \rho = h + \mathfrak{g} $, which satisfies that $ \mathfrak{g} \leq \mu - \rho \leq \mu - \mathfrak{g} $. Choose the code $ \mathcal{C}_{\boldsymbol\gamma} \subseteq \mathbb{F}_{q^r}^\mu $ in (\ref{eq vector gamma as code}) to be the $ (\mu-\rho) $-dimensional linear AG code from Lemma \ref{lemma goppa}. Then the field of linearity of the MSRD code $ \mathcal{C}_k(\mathbf{a}, \boldsymbol\beta) \subseteq \mathbb{F}_{q^m}^{g r} $ has size
\begin{equation}
|\mathbb{F}_{q^m}| = q^{r \rho} = \left( q^r \right)^{h + \mathfrak{g}} .
\label{eq field size 12}
\end{equation}
\end{theorem}
%

We next particularize Theorem \ref{th using ag codes} to three concrete choices of the curve $ \mathcal{X} $ to construct linear MSRD codes with smaller field sizes.

\subsection{Using Hermitian AG codes} \label{subsec plugging an Algebraic-Geometry code Hermitian}

We start by exploring \textit{Hermitian curves} $ \mathcal{X} $ (see \cite[Sec. 8.3]{stichtenothbook}). Fix a positive integer $ s $ such that $ q^r = p^{2s} $, where $ p $ is prime. The Hermitian curve is the projective plane curve with homogeneous equation
$$ x^{q^{ \frac{r}{2} } + 1} - y^{q^{ \frac{r}{2} }} z - y z^{q^{ \frac{r}{2} }} = 0. $$
This curve has $ q^{ \frac{3r}{2} }+1 $ rational points and genus 
$$ \mathfrak{g} = \frac{q^{ \frac{r}{2} } \left( q^{ \frac{r}{2} }-1 \right) }{2} . $$
Therefore, we may choose $ \mu = q^{3r/2} $ in Theorem \ref{th using ag codes}, and we deduce the following.

\begin{corollary} \label{cor using hermitian codes}
Let the notation and assumptions be as in Theorem \ref{th using ag codes}, but where $ \mathcal{X} $ is the Hermitian curve above, and where $ \mu = q^{ 3r/2 } $. 
Then the field of linearity of the MSRD code $ \mathcal{C}_k(\mathbf{a}, \boldsymbol\beta) \subseteq \mathbb{F}_{q^m}^{g r} $ has size
\begin{equation}
|\mathbb{F}_{q^m}| = \left( q^r \right)^{h + \mathfrak{g}} = \mu^{ \frac{2}{3} (h + \mathfrak{g}) } = \mu^{ \frac{1}{3} ( 2h + \mu^{2/3} - \mu^{1/3} ) } ,
\label{eq field size 13}
\end{equation}
that is, $ m = r \left( h + \frac{1}{2} \left( \mu^{ \frac{2}{3} } - \mu^{ \frac{1}{3} } \right) \right) $, where $ \mu = \frac{g}{q-1} $.
\end{corollary}
%
%

\subsection{Using Suzuki AG codes} \label{subsec plugging an Algebraic-Geometry code Suzuki}

In this subsection, we explore \textit{Suzuki} curves $ \mathcal{X} $ (see \cite{suzuki}). Let $ r $ and $ s $ be positive integers such that $ r $ divides $ 2s + 1 $, and consider $ q^r = 2^{2s + 1} $. The Suzuki curve is the projective plane curve with homogeneous equation
$$ x^{2^s} \left( y^{q^r} + y x^{q^r - 1} \right) = z^{2^s} \left( z^{q^r} + z x^{q^r - 1} \right) . $$
This curve has $ q^{2r} + 1 $ rational points over $ \mathbb{F}_{q^r} $ by \cite[Prop. 2.1]{suzuki}, and genus $ \mathfrak{g} = 2^s \left( q^r - 1 \right) $ by \cite[Lemma 1.9]{suzuki}. Therefore, we may choose $ \mu = q^{2r} $ in Theorem \ref{th using ag codes}, hence 
$$ \mathfrak{g} = 2^s \left( \mu^{ \frac{1}{2} } - 1 \right) \leq \mu^{ \frac{1}{4} } \left( \mu^{ \frac{1}{2} } - 1 \right) = \mu^{ \frac{3}{4} } - \mu^{ \frac{1}{4} } , $$
and we deduce the following consequence.

\begin{corollary} \label{cor using suzuki codes}
Let the notation and assumptions be as in Theorem \ref{th using ag codes}, but where $ \mathcal{X} $ is the Suzuki curve above, and where $ \mu = q^{ 2r } $. 
Then the field of linearity of the MSRD code $ \mathcal{C}_k(\mathbf{a}, \boldsymbol\beta) \subseteq \mathbb{F}_{q^m}^{g r} $ has size
\begin{equation}
|\mathbb{F}_{q^m}| = \left( q^r \right)^{h + \mathfrak{g}} = \mu^{ \frac{1}{2} (h + \mathfrak{g}) } \leq \mu^{ \frac{1}{2} \left( h + \mu^{ \frac{3}{4} } - \mu^{ \frac{1}{4} } \right) } ,
\label{eq field size 15}
\end{equation}
that is, $ m = r \left( h + 2^s \left( \mu^{ \frac{1}{2} } - 1 \right) \right) \leq r \left( h + \mu^{ \frac{3}{4} } - \mu^{ \frac{1}{4} } \right) $, where $ \mu = \frac{g}{q-1} $.
\end{corollary}
%

\subsection{Using Garc{\'i}a-Stichtenoth's AG codes} \label{subsec plugging an Algebraic-Geometry code Garcia-Sticht} 

In this subsection, we explore the second sequence of curves $ (\mathcal{X}_i)_{i=1}^\infty $ given by Garc{\'i}a and Stichtenoth (see \cite{garciatower2} or \cite[Sec. 7.4]{stichtenothbook}). Fix a positive integer $ s $ such that $ q^r = p^{2s} $, where $ p $ is prime. For these curves, it is more convenient to define recursively the associated sequence of algebraic function fields $ (\mathcal{F}_i)_{i=1}^\infty $. First we define $ \mathcal{F}_1 = \mathbb{F}_{q^r}(x_1) $, where $ x_1 $ is transcendental over $ \mathbb{F}_{q^r} $, and then we define recursively $ \mathcal{F}_{i+1} = \mathcal{F}_i(x_{i+1}) $, where $ x_{i+1} $ is algebraic over $ \mathcal{F}_i $ satisfying the equation
$$ x_{i+1}^{q^{ \frac{r}{2} }} + x_{i+1} = \frac{x_i^{q^{ \frac{r}{2} }}}{x_i^{q^{ \frac{r}{2} } - 1} + 1 } , $$
for all $ i \in \mathbb{N}_+ $. The $ i $th curve $ \mathcal{X}_i $ has $ q^{ \frac{ir}{2} } \left( q^{ \frac{r}{2} } -1 \right) + 1 $ rational points, and its genus is
\begin{displaymath}
\mathfrak{g} (\mathcal{X}_i) = \left\lbrace 
\begin{array}{ll}
\left( q^{\frac{i r}{4}} - 1 \right) ^2 & \textrm{if } i \textrm{ is even}, \\
\left( q^{\frac{(i+1) r}{4}} - 1 \right) \cdot \left( q^{\frac{(i-1)r}{4}} - 1 \right) & \textrm{if } i \textrm{ is odd}, \\
\end{array} \right.
\end{displaymath}
by \cite[Remark 3.8]{garciatower2}. 
We will choose $ \mu_i = q^{ \frac{ir}{2} } \left( q^{ \frac{r}{2} } -1 \right) $ in Theorem \ref{th using ag codes}, hence 
$$ \mathfrak{g}_i = \mathfrak{g}(\mathcal{X}_i) \leq q^{ \frac{ir}{2} } = \frac{\mu_i}{ q^{ \frac{r}{2} } - 1 } , $$
for all $ i \in \mathbb{N}_+ $, and we deduce the following consequence.

\begin{corollary} \label{cor using garcia codes}
Let the notation and assumptions be as in Theorem \ref{th using ag codes}, but where $ \mathcal{X}_i $ and $ \mu_i $ are as above, for $ i \in \mathbb{N}_+ $. 
Then, for $ i \in \mathbb{N}_+ $, the field of linearity of the MSRD code $ \mathcal{C}_{k_i}(\mathbf{a}_i, \boldsymbol\beta_i) \subseteq \mathbb{F}_{q^{m_i}}^{g_i r} $ has size
\begin{equation}
|\mathbb{F}_{q^{m_i}}| = \left( q^r \right)^{h_i + \mathfrak{g}_i} \leq \left( q^r \right)^{ h_i + q^{ \frac{ir}{2} } } = \left( \frac{\mu_i}{ q^{ \frac{r}{2} } - 1 } \right)^{ \frac{2}{i} \left( h_i + \frac{\mu_i}{ q^{ \frac{r}{2} } - 1 } \right) } .
\label{eq field size 17}
\end{equation}
that is, $ m_i \leq r \left( h_i + q^{ \frac{ir}{2} } \right) = r \left( h_i + \frac{\mu_i}{ q^{ \frac{r}{2} } - 1 } \right) $, where $ \mu_i = \frac{g_i}{q-1} $.
\end{corollary}
%

\section{Summary of results and comparisons} \label{sec summary and final considerations}

In this final section, we will summarize the parameters of the MSRD codes and PMDS codes obtained throughout this work, and compare them to those from the literature.

The parameters of the MSRD codes obtained in Subsections \ref{subsec plugging a basis}, \ref{subsec plugging an MDS code}, \ref{subsec plugging a Hamming code}, \ref{subsec plugging a long BCH code}, \ref{subsec plugging an Algebraic-Geometry code Hermitian}, \ref{subsec plugging an Algebraic-Geometry code Suzuki}, and \ref{subsec plugging an Algebraic-Geometry code Garcia-Sticht} are summarized in Table \ref{table comparisons MSRD}. The parameters of the PMDS codes obtained in Subsections \ref{subsec plugging a basis}, \ref{subsec plugging an MDS code}, and \ref{subsec plugging a long BCH code} are summarized in Table \ref{table comparisons PMDS}.

\subsection{Comparison among MSRD codes}

\begin{table} [!h]
\centering
\begin{tabular}{c||c|c|c}
\hline
&&\\[-0.8em]
Code $ \mathcal{C}_{\boldsymbol\gamma} $ & $ q $, $ r $, $ h $ & No. matrix sets $ g $ & Field of linearity $ q^m $ \\[0.3em]
\hline\hline
&&\\[-0.8em]
Trivial $ \mathcal{C}_{\boldsymbol\gamma} = \{ 0 \} $ & Any & $ q-1 $ & $ q^r = (g+1)^r $, $ m = r $ \\[0.3em]
\hline 
&&\\[-0.8em]
MDS & Any & $ (q-1) \left( q^r + 1 \right) $ & $ \left( \frac{g}{q-1} -1 \right)^{\min \left\lbrace h, \frac{g}{q-1} \right\rbrace } $ \\[0.3em]
\hline 
&&\\[-0.8em]
Hamming, $ \rho \in \mathbb{N}_+ $ & $ h = 2 $ & $ (q-1) \cdot \frac{q^{r \rho} - 1}{q^r - 1} $ & $ q^{r \rho} = \frac{q^r - 1}{q-1} \cdot g + 1 $ \\[0.3em]
\hline 
&&\\[-0.8em]
Pr. BCH, $ s \in \mathbb{N}_+ $ & Any & $ (q-1) \left( q^{rs} - 1 \right) $ & $ \leq q^r \cdot \left( \frac{g}{q-1} + 1 \right) ^{ \left\lceil \frac{q^r - 1}{q^r} (h -1) \right\rceil} $ \\[0.3em]
\hline 
&&\\[-0.8em]
Hermitian AG & $ q^r = p^{2s} $ & $ (q-1) q^{ \frac{3r}{2} } $ & $ \mu^{ \frac{1}{3} ( 2h + \mu^{2/3} - \mu^{1/3} ) } $, $ \mu = \frac{g}{q-1} $ \\[0.3em]
\hline 
&&\\[-0.8em]
Suzuki AG & $ q^r = 2^{2s+1} $ & $ (q-1) q^{2r} $ & $ \leq \mu^{ \frac{1}{2} \left( h + \mu^{3/4} - \mu^{1/4} \right) } $, $ \mu = \frac{g}{q-1} $ \\[0.3em]
\hline 
&&\\[-0.8em]
AG \cite{garciatower2}, $ i \in \mathbb{N}_+ $ & $ q^r = p^{2s} $ & $ (q-1) \left( q^{ \frac{r}{2} } -1 \right) q^{ \frac{ir}{2} } $ & $ \leq \left( \frac{\mu_i}{ q^{ \frac{r}{2} } - 1 } \right)^{ \frac{2}{i} \left( h_i + \frac{\mu_i}{ q^{ \frac{r}{2} } - 1 } \right) } $, $ \mu_i = \frac{g_i}{q-1} $ \\[0.3em]
\hline 
\end{tabular}
\caption{Table summarizing the code parameters of the linear MSRD codes obtained in this work throughout Subsections \ref{subsec plugging a basis}, \ref{subsec plugging an MDS code}, \ref{subsec plugging a Hamming code}, \ref{subsec plugging a long BCH code}, \ref{subsec plugging an Algebraic-Geometry code Hermitian}, \ref{subsec plugging an Algebraic-Geometry code Suzuki}, and \ref{subsec plugging an Algebraic-Geometry code Garcia-Sticht}. They are $ \mathbb{F}_{q^m} $-linear codes in $ \mathbb{F}_{q^m}^{gr} $ with code length $ N = gr $, dimension $ k = gr - h $, and minimum sum-rank distance $ d = h+1 $. Their codewords can be seen as lists of $ g $ matrices over $ \mathbb{F}_q $ of size $ m \times r $, where $ m = r \rho $, $ \rho \in \mathbb{N}_+ $. The linear MSRD code in the first row was obtained in \cite{linearizedRS}, and later independently in \cite{caruso} and \cite{tovo-MSRD}. The codes in the other six rows are new. The symbol $ \leq $ in the last column implies that the given expression is an upper bound on the smallest field size $ q^m $ possible for the corresponding codes. }
\label{table comparisons MSRD}
\end{table}

We start by discussing MSRD codes. First of all, smaller values of $ g $ and $ r $ in Table \ref{table comparisons MSRD} may be obtained in each row, while keeping the same values of $ q $, $ m $ and $ h $, by puncturing or shortening the corresponding MSRD codes, see \cite[Cor. 7]{gsrws}.

As discussed in the Introduction, any MRD code \cite{delsartebilinear, gabidulin, roth} is an MSRD code, however, their fields of linearity have size $ q^m \geq q^{gr} \geq 2^{gr} $, thus exponential in the code length $ N = gr $. All of the field sizes $ q^m $ in Table \ref{table comparisons MSRD} are much smaller than $ 2^{gr} $. 

The first known construction of MSRD codes with sub-exponential field sizes is that of linearized Reed-Solomon codes, introduced in \cite{linearizedRS}, and later independently in \cite{caruso, tovo-MSRD}. They correspond to the first row in Table \ref{table comparisons MSRD}. The other six rows in Table \ref{table comparisons MSRD} correspond to new MSRD codes, all of which attain larger values of $ g $ (relative to the other parameters) than linearized Reed-Solomon codes, which require $ g < q $. We next show that each of these six MSRD codes attains the smallest field sizes $ q^m $ for some parameter regime.

The field sizes $ q^m $ of the MSRD codes in row 3 are the smallest among all MSRD codes relative to $ q $, $ r $ and $ g $, in the sense that they meet the bound (\ref{eq bound 2 alberto refined}). However, these codes are restricted to $ h = 2 $ (dimension or co-dimension $ 2 $). 

The MSRD codes in rows 2 and 4 in Table \ref{table comparisons MSRD} attain smaller field sizes $ q^m $ than linearized Reed-Solomon codes for small $ h $ relative to $ r $, by looking at the exponents (the base is roughly $ g $ in rows 1, 2 and 4). Consider the parameter regime $ q-1 < g \leq (q-1) (q^r-1) $ (i.e., setting $ s=1 $ in row 4), which are unattainable values of $ g $ for linearized Reed-Solomon codes. In this regime, the MSRD codes in row 2 attain smaller field sizes $ q^m $ than those in row 4 if 
$$ h \geq \frac{(q^r + 1)^2}{q^r-1} = \Omega(q^r) $$
(for code length $ N = gr = (q-1)(q^r-1)r = \Theta(q^{r+1}r) $, the relative co-dimension is thus required to be $ h/N = \Omega (q^{-1}r^{-1}) $). When $ h < q^r $ (thus relative co-dimension $ h/N < (q-1)^{-1}r^{-1} $), the MSRD codes in row 4 attain smaller field sizes than those in row 2. Notice that in this parameter regime, i.e., $ q-1 < g \leq (q-1) (q^r-1) $, the MSRD codes in rows 5, 6 and 7 require larger field sizes than those in rows 1, 2, 3 or 4.

Now consider the parameter regime $ g > (q-1)(q^r+1) $ and $ h > 2 $, thus we may only consider the codes in the last four rows of Table \ref{table comparisons MSRD}. Consider first the MSRD codes in row 5, assume $ h < \mu $ and let $ \varepsilon > 0 $ be such that $ \varepsilon h > \mu^{2/3} - \mu^{1/3} $. Then the field sizes in row 5 are $ q^m = \mu^{h (2+\varepsilon)/3} $. Since only $ \varepsilon \mu > \varepsilon h > \mu^{2/3} - \mu^{1/3} $ is required, we may consider $ \varepsilon > 0 $ to be as small as desired for large enough $ \mu $ and $ h $. In that case, the field sizes in row 5 are $ \mu^{\mathcal{O}(2h/3)} $, which are asymptotically smaller than the field sizes $ \mu^{\Theta(h)} $ from row 4. Consider now the MSRD codes in row 6, assume $ h < \mu $ and let $ \varepsilon > 0 $ be such that $ \varepsilon h > \mu^{3/4} - \mu^{1/4} $. Again, $ \varepsilon > 0 $ may be as small as desired for large enough $ h $ and $ \mu $. Then, the field sizes in row 6 are $ \mu^{\mathcal{O}(h/2)} $, which are asymptotically smaller than the field sizes $ \mu^{\mathcal{O}(2h/3)} $ from row 5, however with stronger restrictions, as we require $ \varepsilon h > \mu^{3/4} - \mu^{1/4} $ instead of $ \varepsilon h > \mu^{2/3} - \mu^{1/3} $. Finally, consider the MSRD codes in row 7 and assume that $ \mu/(q^{r/2}-1) < h < \mu $. Then the field sizes in row 7 are $ \mu^{\mathcal{O}(h/i)} $, where $ i \in \mathbb{N}_+ $ may be arbitrary. Such field sizes are asymptotically smaller than the field sizes from rows 5 and 6, however with stronger restrictions, as we require $ h > \mu/(q^{r/2}-1) $ instead of $ \varepsilon h > \mu^{2/3} - \mu^{1/3} $ or $ \varepsilon h > \mu^{3/4} - \mu^{1/4} $. 
%

Finally, we mention two recent works where new MSRD codes are constructed. Twisted linearized Reed-Solomon codes have been introduced in \cite{neri-twisted}
. However, their field sizes are not smaller than those of linearized Reed-Solomon codes (row 1 in Table \ref{table comparisons MSRD}). Some constructions of MSRD codes were recently given in \cite{alberto-fundamental}. However, such codes are only $ \mathbb{F}_q $-linear, and have minimum sum-rank distance equal to $ 2 $ or $ \sum_{i=1}^g r_i - 1 $ (total number of columns, across all matrices, minus $ 1 $), or require the number of rows or columns to be $ 1 $ at certain positions in the matrices in (\ref{eq codeword as list of matrices}). 
%

\subsection{Comparison among PMDS codes}

\begin{table} [!h]
\centering
\begin{tabular}{c||c|c}
\hline
&\\[-0.8em]
Code $ \mathcal{C}_{\boldsymbol\gamma} $ & Restrictions on $ r $, $ \delta $, $ g $, $ h $, $ q $ & Field size $ q^m $ \\[0.3em]
\hline\hline
&\\[-0.8em]
Trivial $ \mathcal{C}_{\boldsymbol\gamma} = \{ 0 \} $ & $ \max \{ \nu, g \} < q \leq 2 \max \{ \nu, g \} $ & $ \leq (2 \max \{ \nu , g \})^r $ \\[0.3em]
\hline 
&\\[-0.8em]
MDS & $ g = (q-1) \left( q^r + 1 \right) $ or $ (2 \nu)^r > \frac{g}{\nu} $ & $ \leq \max \left\lbrace \left( 2 \nu \right)^r , \left\lfloor \frac{g}{\nu} \right\rfloor - 1 \right\rbrace ^{\min \left\lbrace h,  \left\lfloor \frac{g}{\nu} \right\rfloor \right\rbrace } $ \\[0.3em]
\hline 
&\\[-0.8em]
Primitive BCH & $ g = (q-1) \left( q^{rs} - 1 \right) $ and $ q > \nu $ & $ \leq (2 \nu)^r \cdot \left( \left\lfloor \frac{g}{\nu} \right\rfloor + 1 \right) ^{ h -1} $ \\[0.3em]
\hline  
\end{tabular}
\caption{Table summarizing the code parameters of the linear PMDS codes obtained in this work throughout Subsections \ref{subsec plugging a basis}, \ref{subsec plugging an MDS code}, and \ref{subsec plugging a long BCH code}. They are $ \mathbb{F}_{q^m} $-linear codes in $ \mathbb{F}_{q^m}^{g \nu} $, where $ r $ is the locality, $ \delta $ is the local distance, $ g $ is the number of local sets, $ h $ is the number of global parities, $ \nu = r+\delta -1 $ is the local-set size and $ q $ is a power of $ 2 $. The field size of the local codes is a subfield of $ \mathbb{F}_q $. The linear PMDS code in the first row was obtained in \cite{universal-lrc}. The codes in the other two rows are new. The symbol $ \leq $ in the last column implies that the given expression is an upper bound on the smallest field size $ q^m $ possible for the corresponding codes. }
\label{table comparisons PMDS}
\end{table}

We now turn to discussing PMDS codes. For $ h \in \{ 0,1,2,3 \} $, the PMDS codes from \cite{blaum-RAID, gopi} have smaller field sizes than the PMDS codes in Table \ref{table comparisons PMDS}. For dimension $ r+1 $, the PMDS codes from \cite{zitan} also have smaller field sizes than those in Table \ref{table comparisons PMDS}.

To the best of our knowledge, the PMDS codes for general parameters with the smallest known field sizes are those in \cite{cai-field, gopi-field, universal-lrc}, being those from \cite{universal-lrc} exactly the PMDS codes in row 1 in Table \ref{table comparisons PMDS}. The constructions from \cite{cai-field, gopi-field, universal-lrc}, put together, yield PMDS codes with field sizes $ q^m $ such that
\begin{equation}
\max \{ \nu, g \}^{\min \{ r,h \}} < q^m \leq (2 \max \{ \nu, g \})^{\min \{ r,h \}}.
\label{eq field sizes best}
\end{equation}
If $ g > (2\nu)^r \nu $, then the field sizes of the PMDS codes in row 2 in Table \ref{table comparisons PMDS} would be
$$ q^m \leq \left( \left\lfloor \frac{g}{\nu} \right\rfloor - 1 \right)^h, $$
which would be strictly smaller than those in row 3 in Table \ref{table comparisons PMDS} and in (\ref{eq field sizes best}) when $ h < r $. 

Asymptotically, if $ \nu = \mathcal{O}(1) $ and $ g > \nu $ grows unboundedly, the field sizes of the PMDS codes in row 3 in Table \ref{table comparisons PMDS} are
$$ q^m = \mathcal{O} \left( \left( \frac{g}{\nu} \right)^{h-1} \right), $$
which would be asymptotically smaller than those of the PMDS codes in row 2 in Table \ref{table comparisons PMDS} and those in (\ref{eq field sizes best}) when $ h < r $.

{\footnotesize

\bibliographystyle{plain}

}

\newpage

\appendix

\section*{Appendix: Tables with even field sizes for MSRD codes}

In this appendix, we provide several tables of attainable field sizes $ q^m $, divisible by $ 2 $, among the linear MSRD codes obtained in this work.

First, we give a summary in Table \ref{table fields MSRD q = general}, which is similar to Table \ref{table comparisons MSRD}, but where the field size $ q^m $ is not compared to $ g $, but written as a function of $ q $, $ r $ and $ h $, excluding $ g $. The reason behind this is that typically the maximum attainable value of $ g $ is quite large for most of these codes, and in most cases we would puncture them in order to have a much smaller number of matrix sets $ g $. 

In Tables \ref{table fields MSRD g >= 7}, \ref{table fields MSRD g >= 15} and \ref{table fields MSRD g >= 31}, we fix $ g $ and let other parameters vary. In contrast, in Tables \ref{table fields MSRD N >= 30} and \ref{table fields MSRD N >= 62}, we fix the code length $ N = gr $ and let other parameters vary. In these tables, bold numbers indicate field sizes that are the smallest among MSRD codes of the same parameters. As linearized Reed-Solomon codes have the same field sizes for all $ h $, a bold number in that row means that the field size is the smallest for the corresponding parameters for some $ h $.

The field sizes attained by linear MSRD codes based on AG codes (Subsections \ref{subsec plugging an Algebraic-Geometry code Hermitian}, \ref{subsec plugging an Algebraic-Geometry code Suzuki} and \ref{subsec plugging an Algebraic-Geometry code Garcia-Sticht}) are quite larger than those obtained by the other linear MSRD codes for small parameters. In general, MSRD codes based on AG codes (as in Subsection \ref{subsec plugging an Algebraic-Geometry code}) are mostly of asymptotic interest. For this reason, they are not included in Tables \ref{table fields MSRD q = 2}, \ref{table fields MSRD g >= 7}, \ref{table fields MSRD g >= 15}, \ref{table fields MSRD g >= 31}, \ref{table fields MSRD N >= 30} and \ref{table fields MSRD N >= 62}.
%

Finally, at the end of each table we consider the smallest field size attainable by an MRD code for the corresponding parameters. As it can be seen, MRD codes always require significant larger field sizes than the MSRD codes from this work, for the same parameters.

\begin{table} [h!]
\centering
\begin{tabular}{c||c|c|c}
\hline
&&\\[-0.8em]
Code $ \mathcal{C}_{\boldsymbol\gamma} $ & $ q $, $ r $, $ h $ & No. matrix sets $ g $ & Field of linearity $ q^m $ \\[0.3em]
\hline\hline
&&\\[-0.8em]
Trivial $ \mathcal{C}_{\boldsymbol\gamma} = \{ 0 \} $ (Lin. RS) & Any & $ q-1 $ & $ q^r $ \\[0.3em]
\hline 
&&\\[-0.8em]
MDS & Any & $ (q-1) \left( q^r + 1 \right) $ & $ q^{r \min \left\lbrace h, q^r + 1 \right\rbrace } $ \\[0.3em]
\hline 
&&\\[-0.8em]
Hamming, $ \rho \in \mathbb{N}_+ $ & $ h = 2 $ & $ (q-1) \cdot \frac{q^{r \rho} - 1}{q^r - 1} $ & $ q^{r \rho} $ \\[0.3em]
\hline 
&&\\[-0.8em]
Pr. BCH, $ s \in \mathbb{N}_+ $ & Any & $ (q-1) \left( q^{rs} - 1 \right) $ & $ \leq q^{r \left( 1 + s \left\lceil \frac{q^r - 1}{q^r} (h-1) \right\rceil \right) } $ \\[0.3em]
\hline 
&&\\[-0.8em]
Hermitian AG & $ q^r = p^{2s} $ & $ (q-1) q^{ \frac{3r}{2} } $ & $ q^{r \left( h + \frac{1}{2} \left( q^r - q^{ \frac{r}{2} } \right) \right) } $ \\[0.3em]
\hline 
&&\\[-0.8em]
Suzuki AG & $ q^r = 2^{2s+1} $ & $ (q-1) q^{2r} $ & $ q^{r \left( h + 2^s \left( q^r - 1 \right) \right) } $ \\[0.3em]
\hline 
&&\\[-0.8em]
AG \cite{garciatower2}, $ i \in \mathbb{N}_+ $ & $ q^r = p^{2s} $ & $ (q-1) \left( q^{ \frac{r}{2} } -1 \right) q^{ \frac{ir}{2} } $ & $ \leq q^{r \left( h_i + q^{ \frac{ir}{2} } \right) } $ \\[0.3em]
\hline 
\end{tabular}
\caption{Table summarizing the code parameters of the linear MSRD codes obtained in this work. In contrast with Table \ref{table comparisons MSRD}, field sizes are described in terms of $ q $, $ r $ and $ h $, excluding $ g $. }
\label{table fields MSRD q = general}
\end{table}

\begin{table} [t]
\centering
\begin{tabular}{c||cc|cc|cc|cc|cc}
\hline
&&&&&&&&\\[-0.8em]
 Table for $ q = 2 $ & \multicolumn{2}{|c|}{$ r = 2 $} & \multicolumn{2}{|c|}{$ r = 3 $} & \multicolumn{2}{|c|}{$ r = 4 $} & \multicolumn{2}{|c|}{$ r = 5 $} & \multicolumn{2}{c}{$ r = 6 $} \\[0.3em]
&&&&&&&&\\[-0.8em]
Code $ \mathcal{C}_{\boldsymbol\gamma} $ & $ 2^m $ & $ g $ & $ 2^m $ & $ g $ & $ 2^m $ & $ g $ & $ 2^m $ & $ g $ & $ 2^m $ & $ g $ \\[0.3em]
\hline\hline
&&&&&&&&\\[-0.8em]
Trivial $ \mathcal{C}_{\boldsymbol\gamma} = \{ 0 \} $ (Lin. RS) & $ 2^2 $ & $ 1 $ & $ 2^3 $ & $ 1 $ & $ 2^4 $ & $ 1 $ & $ 2^5 $ & $ 1 $ & $ 2^6 $ & $ 1 $ \\[0.3em]
\hline 
&&&&&&&&\\[-0.8em]
MDS, $ h = 2 $ & $ 2^4 $ & \multirow{3}{*}{$ 5 $} & $ 2^6 $ & \multirow{3}{*}{$ 9 $} & $ 2^8 $ & \multirow{3}{*}{$ 17 $} & $ 2^{10} $ &  \multirow{3}{*}{$ 33 $} & $ 2^{12} $ &  \multirow{3}{*}{$ 65 $} \\[0.3em]

&&&&&&&&\\[-0.8em]
\phantom{MDS,} $ h = 3 $ & $ 2^6 $ & & $ 2^9 $ & & $ 2^{12} $ & & $ 2^{15} $ & & $ 2^{18} $ & \\[0.3em]

&&&&&&&&\\[-0.8em]
\phantom{MDS,} $ h = 4 $ & $ 2^8 $ & & $ 2^{12} $ & & $ 2^{16} $ & & $ 2^{20} $ & & $ 2^{24} $ & \\[0.3em]
\hline
&&&&&&&&\\[-0.8em]
Hamming, $ \rho = 3 $, $ h = 2 $ & $ 2^6 $ & $ 21 $ & $ 2^9 $ & $ 73 $ & $ 2^{12} $ & $ 273 $ & $ 2^{15} $ & $ 1057 $ & $ 2^{18} $ & $ 4161 $ \\[0.3em]
\hline 
&&&&&&&&\\[-0.8em]
Pr. BCH, $ s = 2 $, $ h = 2 $ & $ 2^6 $ & \multirow{3}{*}{$ 15 $} & $ 2^9 $ & \multirow{3}{*}{$ 63 $} & $ 2^{12} $ & \multirow{3}{*}{$ 255 $} & $ 2^{15} $ & \multirow{3}{*}{$ 1023 $} & $ 2^{18} $ & \multirow{3}{*}{$ 4095 $} \\[0.3em]

&&&&&&&&\\[-0.8em]
\phantom{Pr. BCH, $ s = 2 $,} $ h = 3 $ & $ 2^{10} $ & & $ 2^{15} $ & & $ 2^{20} $ & & $ 2^{25} $ & & $ 2^{30} $ & \\[0.3em]

&&&&&&&&\\[-0.8em]
\phantom{Pr. BCH, $ s = 2 $,} $ h = 4 $ & $ 2^{14} $ & & $ 2^{21} $ & & $ 2^{28} $ & & $ 2^{35} $ & & $ 2^{42} $ & \\[0.3em]
\hline 
&&&&&&&&\\[-0.8em]
Pr. BCH, $ s = 3 $, $ h = 2 $ & $ 2^8 $ & \multirow{3}{*}{$ 63 $} & $ 2^{12} $ & \multirow{3}{*}{$ 511 $} & $ 2^{16} $ & \multirow{3}{*}{$ 4095 $} & $ 2^{20} $ & \multirow{3}{*}{$ 2^{15}-1 $} & $ 2^{24} $ & \multirow{3}{*}{$ 2^{18}-1 $} \\[0.3em]

&&&&&&&&\\[-0.8em]
\phantom{Pr. BCH, $ s = 3 $,} $ h = 3 $ & $ 2^{14} $ & & $ 2^{21} $ & & $ 2^{28} $ & & $ 2^{35} $ & & $ 2^{42} $ & \\[0.3em]

&&&&&&&&\\[-0.8em]
\phantom{Pr. BCH, $ s = 3 $,} $ h = 4 $ & $ 2^{20} $ & & $ 2^{30} $ & & $ 2^{40} $ & & $ 2^{50} $ & & $ 2^{60} $ & \\[0.3em]
\hline 
\end{tabular}
\caption{Table for fixed $ q = 2 $, while other parameters vary. }
\label{table fields MSRD q = 2}
\end{table}

\begin{table} [t]
\centering
\begin{tabular}{c||cc|cc|cc|cc|cc}
\hline
&&&&&&&&\\[-0.8em]
 Table for $ g = 7 $, $ q $ even & \multicolumn{2}{|c|}{$ r = 2 $} & \multicolumn{2}{|c|}{$ r = 3 $} & \multicolumn{2}{|c|}{$ r = 4 $} & \multicolumn{2}{|c|}{$ r = 5 $} & \multicolumn{2}{c}{$ r = 6 $} \\[0.3em]
&&&&&&&&\\[-0.8em]
Code $ \mathcal{C}_{\boldsymbol\gamma} $ & $ q^m $ & $ q $ & $ q^m $ & $ q $ & $ q^m $ & $ q $ & $ q^m $ & $ q $ & $ q^m $ & $ q $ \\[0.3em]
\hline\hline
&&&&&&&&\\[-0.8em]
Trivial $ \mathcal{C}_{\boldsymbol\gamma} = \{ 0 \} $ (Lin. RS), $ \mathbf{\forall} h \geq 1 $ & $ \mathbf{2^6} $ & $ 2^3 $ & $ \mathbf{2^{9}} $ & $ 2^3 $ & $ \mathbf{2^{12}} $ & $ 2^3 $ & $ \mathbf{2^{15}} $ & $ 2^3 $ & $ \mathbf{2^{18}} $ & $ 2^3 $ \\[0.3em]
\hline 
&&&&&&&&\\[-0.8em]
MDS, $ h = 2 $ & $ 2^8 $ & \multirow{3}{*}{$ 2^2 $} & $ \mathbf{2^{6}} $ & \multirow{3}{*}{$ 2 $} & $ \mathbf{2^8} $ & \multirow{3}{*}{$ 2 $} & $ \mathbf{2^{10}} $ &  \multirow{3}{*}{$ 2 $} & $ \mathbf{2^{12}} $ &  \multirow{3}{*}{$ 2 $} \\[0.3em]

&&&&&&&&\\[-0.8em]
\phantom{MDS,} $ h = 3 $ & $ 2^{12} $ & & $ \mathbf{2^{9}} $ & & $ \mathbf{2^{12}} $ & & $ \mathbf{2^{15}} $ & & $ \mathbf{2^{18}} $ & \\[0.3em]

&&&&&&&&\\[-0.8em]
\phantom{MDS,} $ h = 4 $ & $ 2^{16} $ & & $ 2^{12} $ & & $ 2^{16} $ & & $ 2^{20} $ & & $ 2^{24} $ & \\[0.3em]
\hline
&&&&&&&&\\[-0.8em]
Hamming, $ \rho = 3 $, $ h = 2 $ & $ \mathbf{2^6} $ & $ 2 $ & $ 2^9 $ & $ 2 $ & $ 2^{12} $ & $ 2 $ & $ 2^{15} $ & $ 2 $ & $ 2^{18} $ & $ 2 $ \\[0.3em]
\hline 
&&&&&&&&\\[-0.8em]
Pr. BCH, $ s = 1,2 $, $ h = 2 $ & $ \mathbf{2^6} $ & \multirow{3}{*}{$ 2 $} & $ \mathbf{2^6} $ & \multirow{3}{*}{$ 2 $} & $ \mathbf{2^8} $ & \multirow{3}{*}{$ 2 $} & $ \mathbf{2^{10}} $ & \multirow{3}{*}{$ 2 $} & $ \mathbf{2^{12}} $ & \multirow{3}{*}{$ 2 $} \\[0.3em]

&&&&&&&&\\[-0.8em]
\phantom{Pr. BCH, $ s = 1,2 $,} $ h = 3 $ & $ 2^{10} $ & & $ \mathbf{2^{9}} $ & & $ \mathbf{2^{12}} $ & & $ \mathbf{2^{15}} $ & & $ \mathbf{2^{18}} $ & \\[0.3em]

&&&&&&&&\\[-0.8em]
\phantom{Pr. BCH, $ s = 1,2 $,} $ h = 4 $ & $ 2^{14} $ & & $ 2^{12} $ & & $ 2^{16} $ & & $ 2^{20} $ & & $ 2^{24} $ & \\[0.3em]
\hline\hline
&&&&&&&&\\[-0.8em]
Best MRD code possible, $ \mathbf{\forall} h \geq 1 $ & $ 2^{14} $ & $ 2 $ & $ 2^{28} $ & $ 2 $ & $ 2^{42} $ & $ 2 $ & $ 2^{56} $ & $ 2 $ & $ 2^{70} $ & $ 2 $ \\[0.3em]
\hline 
\end{tabular}
\caption{Table for fixed $ g = 7 $, while other parameters vary. }
\label{table fields MSRD g >= 7}
\end{table}

\begin{table} [t]
\centering
\begin{tabular}{c||cc|cc|cc|cc|cc}
\hline
&&&&&&&&\\[-0.8em]
 Table for $ g = 15 $, $ q $ even & \multicolumn{2}{|c|}{$ r = 2 $} & \multicolumn{2}{|c|}{$ r = 3 $} & \multicolumn{2}{|c|}{$ r = 4 $} & \multicolumn{2}{|c|}{$ r = 5 $} & \multicolumn{2}{c}{$ r = 6 $} \\[0.3em]
&&&&&&&&\\[-0.8em]
Code $ \mathcal{C}_{\boldsymbol\gamma} $ & $ q^m $ & $ q $ & $ q^m $ & $ q $ & $ q^m $ & $ q $ & $ q^m $ & $ q $ & $ q^m $ & $ q $ \\[0.3em]
\hline\hline
&&&&&&&&\\[-0.8em]
Trivial $ \mathcal{C}_{\boldsymbol\gamma} = \{ 0 \} $ (Lin. RS), $ \mathbf{\forall} h \geq 1 $ & $ \mathbf{2^8} $ & $ 2^4 $ & $ \mathbf{2^{12}} $ & $ 2^4 $ & $ \mathbf{2^{16}} $ & $ 2^4 $ & $ \mathbf{2^{20}} $ & $ 2^4 $ & $ \mathbf{2^{24}} $ & $ 2^4 $ \\[0.3em]
\hline 
&&&&&&&&\\[-0.8em]
MDS, $ h = 2 $ & $ 2^8 $ & \multirow{3}{*}{$ 2^2 $} & $ 2^{12} $ & \multirow{3}{*}{$ 2^2 $} & $ \mathbf{2^8} $ & \multirow{3}{*}{$ 2 $} & $ \mathbf{2^{10}} $ &  \multirow{3}{*}{$ 2 $} & $ \mathbf{2^{12}} $ &  \multirow{3}{*}{$ 2 $} \\[0.3em]

&&&&&&&&\\[-0.8em]
\phantom{MDS,} $ h = 3 $ & $ 2^{12} $ & & $ 2^{18} $ & & $ \mathbf{2^{12}} $ & & $ \mathbf{2^{15}} $ & & $ \mathbf{2^{18}} $ & \\[0.3em]

&&&&&&&&\\[-0.8em]
\phantom{MDS,} $ h = 4 $ & $ 2^{16} $ & & $ 2^{24} $ & & $ \mathbf{2^{16}} $ & & $ \mathbf{2^{20}} $ & & $ \mathbf{2^{24}} $ & \\[0.3em]
\hline
&&&&&&&&\\[-0.8em]
Hamming, $ \rho = 3 $, $ h = 2 $ & $ \mathbf{2^6} $ & $ 2 $ & $ \mathbf{2^9} $ & $ 2 $ & $ 2^{12} $ & $ 2 $ & $ 2^{15} $ & $ 2 $ & $ 2^{18} $ & $ 2 $ \\[0.3em]
\hline 
&&&&&&&&\\[-0.8em]
Pr. BCH, $ s = 1,2 $, $ h = 2 $ & $ \mathbf{2^6} $ & \multirow{3}{*}{$ 2 $} & $ \mathbf{2^9} $ & \multirow{3}{*}{$ 2 $} & $ \mathbf{2^8} $ & \multirow{3}{*}{$ 2 $} & $ \mathbf{2^{10}} $ & \multirow{3}{*}{$ 2 $} & $ \mathbf{2^{12}} $ & \multirow{3}{*}{$ 2 $} \\[0.3em]

&&&&&&&&\\[-0.8em]
\phantom{Pr. BCH, $ s = 1,2 $,} $ h = 3 $ & $ 2^{10} $ & & $ 2^{15} $ & & $ \mathbf{2^{12}} $ & & $ \mathbf{2^{15}} $ & & $ \mathbf{2^{18}} $ & \\[0.3em]

&&&&&&&&\\[-0.8em]
\phantom{Pr. BCH, $ s = 1,2 $,} $ h = 4 $ & $ 2^{14} $ & & $ 2^{21} $ & & $ \mathbf{2^{16}} $ & & $ \mathbf{2^{20}} $ & & $ \mathbf{2^{24}} $ & \\[0.3em]
\hline\hline
&&&&&&&&\\[-0.8em]
Best MRD code possible, $ \mathbf{\forall} h \geq 1 $ & $ 2^{30} $ & $ 2 $ & $ 2^{45} $ & $ 2 $ & $ 2^{60} $ & $ 2 $ & $ 2^{75} $ & $ 2 $ & $ 2^{90} $ & $ 2 $ \\[0.3em]
\hline 
\end{tabular}
\caption{Table for fixed $ g = 15 $, while other parameters vary. }
\label{table fields MSRD g >= 15}
\end{table}

\begin{table} [t]
\centering
\begin{tabular}{c||cc|cc|cc|cc|cc}
\hline
&&&&&&&&\\[-0.8em]
 Table for $ g = 31 $, $ q $ even & \multicolumn{2}{|c|}{$ r = 2 $} & \multicolumn{2}{|c|}{$ r = 3 $} & \multicolumn{2}{|c|}{$ r = 4 $} & \multicolumn{2}{|c|}{$ r = 5 $} & \multicolumn{2}{c}{$ r = 6 $} \\[0.3em]
&&&&&&&&\\[-0.8em]
Code $ \mathcal{C}_{\boldsymbol\gamma} $ & $ q^m $ & $ q $ & $ q^m $ & $ q $ & $ q^m $ & $ q $ & $ q^m $ & $ q $ & $ q^m $ & $ q $ \\[0.3em]
\hline\hline
&&&&&&&&\\[-0.8em]
Trivial $ \mathcal{C}_{\boldsymbol\gamma} = \{ 0 \} $ (Lin. RS), $ \mathbf{\forall} h \geq 1 $ & $ \mathbf{2^{10}} $ & $ 2^5 $ & $ \mathbf{2^{15}} $ & $ 2^5 $ & $ \mathbf{2^{20}} $ & $ 2^5 $ & $ \mathbf{2^{25}} $ & $ 2^5 $ & $ \mathbf{2^{30}} $ & $ 2^5 $ \\[0.3em]
\hline 
&&&&&&&&\\[-0.8em]
MDS, $ h = 2 $ & $ \mathbf{2^8} $ & \multirow{3}{*}{$ 2^2 $} & $ 2^{12} $ & \multirow{3}{*}{$ 2^2 $} & $ 2^{16} $ & \multirow{3}{*}{$ 2^2 $} & $ \mathbf{2^{10}} $ &  \multirow{3}{*}{$ 2 $} & $ \mathbf{2^{12}} $ &  \multirow{3}{*}{$ 2 $} \\[0.3em]

&&&&&&&&\\[-0.8em]
\phantom{MDS,} $ h = 3 $ & $ 2^{12} $ & & $ 2^{18} $ & & $ 2^{24} $ & & $ \mathbf{2^{15}} $ & & $ \mathbf{2^{18}} $ & \\[0.3em]

&&&&&&&&\\[-0.8em]
\phantom{MDS,} $ h = 4 $ & $ 2^{16} $ & & $ 2^{24} $ & & $ 2^{32} $ & & $ \mathbf{2^{20}} $ & & $ \mathbf{2^{24}} $ & \\[0.3em]
\hline
&&&&&&&&\\[-0.8em]
Hamming, $ \rho = 3 $, $ h = 2 $ & $ \mathbf{2^8} $ & $ 2 $ & $ \mathbf{2^9} $ & $ 2 $ & $ \mathbf{2^{12}} $ & $ 2 $ & $ 2^{15} $ & $ 2 $ & $ 2^{18} $ & $ 2 $ \\[0.3em]
\hline 
&&&&&&&&\\[-0.8em]
Pr. BCH, $ s = 1,2,3 $, $ h = 2 $ & $ \mathbf{2^8} $ & \multirow{3}{*}{$ 2 $} & $ \mathbf{2^9} $ & \multirow{3}{*}{$ 2 $} & $ \mathbf{2^{12}} $ & \multirow{3}{*}{$ 2 $} & $ \mathbf{2^{10}} $ & \multirow{3}{*}{$ 2 $} & $ \mathbf{2^{12}} $ & \multirow{3}{*}{$ 2 $} \\[0.3em]

&&&&&&&&\\[-0.8em]
\phantom{Pr. BCH, $ s = 1,2,3 $,} $ h = 3 $ & $ 2^{14} $ & & $ \mathbf{2^{15}} $ & & $ \mathbf{2^{20}} $ & & $ \mathbf{2^{15}} $ & & $ \mathbf{2^{18}} $ & \\[0.3em]

&&&&&&&&\\[-0.8em]
\phantom{Pr. BCH, $ s = 1,2,3 $,} $ h = 4 $ & $ 2^{20} $ & & $ 2^{21} $ & & $ 2^{28} $ & & $ \mathbf{2^{20}} $ & & $ \mathbf{2^{24}} $ & \\[0.3em]
\hline\hline
&&&&&&&&\\[-0.8em]
Best MRD code possible, $ \mathbf{\forall} h \geq 1 $ & $ 2^{62} $ & $ 2 $ & $ 2^{93} $ & $ 2 $ & $ 2^{124} $ & $ 2 $ & $ 2^{155} $ & $ 2 $ & $ 2^{186} $ & $ 2 $ \\[0.3em]
\hline 
\end{tabular}
\caption{Table for fixed $ g = 31 $, while other parameters vary. }
\label{table fields MSRD g >= 31}
\end{table}

\begin{table} [t]
\centering
\begin{tabular}{c||cc|cc|cc|cc|cc}
\hline
&&&&&&&&\\[-0.8em]
 Table for $ N = gr = 30 $, $ q $ even & \multicolumn{2}{|c|}{$ r = 2 $} & \multicolumn{2}{|c|}{$ r = 3 $} & \multicolumn{2}{|c|}{$ r = 4 $} & \multicolumn{2}{|c|}{$ r = 5 $} & \multicolumn{2}{c}{$ r = 6 $} \\[0.3em]
&&&&&&&&\\[-0.8em]
Code $ \mathcal{C}_{\boldsymbol\gamma} $ & $ q^m $ & $ q $ & $ q^m $ & $ q $ & $ q^m $ & $ q $ & $ q^m $ & $ q $ & $ q^m $ & $ q $ \\[0.3em]
\hline\hline
&&&&&&&&\\[-0.8em]
Trivial $ \mathcal{C}_{\boldsymbol\gamma} = \{ 0 \} $ (Lin. RS), $ \mathbf{\forall} h \geq 1 $ & $ \mathbf{2^8} $ & $ 2^4 $ & $ \mathbf{2^{12}} $ & $ 2^4 $ & $ \mathbf{2^{16}} $ & $ 2^4 $ & $ \mathbf{2^{15}} $ & $ 2^3 $ & $ \mathbf{2^{18}} $ & $ 2^3 $ \\[0.3em]
\hline 
&&&&&&&&\\[-0.8em]
MDS, $ h = 2 $ & $ 2^8 $ & \multirow{3}{*}{$ 2^2 $} & $ 2^{12} $ & \multirow{3}{*}{$ 2^2 $} & $ \mathbf{2^8} $ & \multirow{3}{*}{$ 2 $} & $ \mathbf{2^{10}} $ &  \multirow{3}{*}{$ 2 $} & $ \mathbf{2^{12}} $ &  \multirow{3}{*}{$ 2 $} \\[0.3em]

&&&&&&&&\\[-0.8em]
\phantom{MDS,} $ h = 3 $ & $ 2^{12} $ & & $ 2^{18} $ & & $ \mathbf{2^{12}} $ & & $ \mathbf{2^{15}} $ & & $ \mathbf{2^{18}} $ & \\[0.3em]

&&&&&&&&\\[-0.8em]
\phantom{MDS,} $ h = 4 $ & $ 2^{16} $ & & $ 2^{24} $ & & $ \mathbf{2^{16}} $ & & $ 2^{20} $ & & $ 2^{24} $ & \\[0.3em]
\hline
&&&&&&&&\\[-0.8em]
Hamming, $ \rho = 3 $, $ h = 2 $ & $ \mathbf{2^6} $ & $ 2 $ & $ \mathbf{2^9} $ & $ 2 $ & $ 2^{12} $ & $ 2 $ & $ 2^{15} $ & $ 2 $ & $ 2^{18} $ & $ 2 $ \\[0.3em]
\hline 
&&&&&&&&\\[-0.8em]
Pr. BCH, $ s = 1,2 $, $ h = 2 $ & $ \mathbf{2^6} $ & \multirow{3}{*}{$ 2 $} & $ \mathbf{2^9} $ & \multirow{3}{*}{$ 2 $} & $ \mathbf{2^8} $ & \multirow{3}{*}{$ 2 $} & $ \mathbf{2^{10}} $ & \multirow{3}{*}{$ 2 $} & $ \mathbf{2^{12}} $ & \multirow{3}{*}{$ 2 $} \\[0.3em]

&&&&&&&&\\[-0.8em]
\phantom{Pr. BCH, $ s = 1,2 $,} $ h = 3 $ & $ 2^{10} $ & & $ 2^{15} $ & & $ \mathbf{2^{12}} $ & & $ \mathbf{2^{15}} $ & & $ \mathbf{2^{18}} $ & \\[0.3em]

&&&&&&&&\\[-0.8em]
\phantom{Pr. BCH, $ s = 1,2 $,} $ h = 4 $ & $ 2^{14} $ & & $ 2^{21} $ & & $ \mathbf{2^{16}} $ & & $ 2^{20} $ & & $ 2^{24} $ & \\[0.3em]
\hline\hline
&&&&&&&&\\[-0.8em]
Best MRD code possible, $ \mathbf{\forall} h \geq 1 $ & $ 2^{30} $ & $ 2 $ & $ 2^{30} $ & $ 2 $ & $ 2^{30} $ & $ 2 $ & $ 2^{30} $ & $ 2 $ & $ 2^{30} $ & $ 2 $ \\[0.3em]
\hline 
\end{tabular}
\caption{Table for fixed $ N = gr = 30 $, while other parameters vary. }
\label{table fields MSRD N >= 30}
\end{table}

\begin{table} [t]
\centering
\begin{tabular}{c||cc|cc|cc|cc|cc}
\hline
&&&&&&&&\\[-0.8em]
  Table for $ N = gr = 62 $, $ q $ even & \multicolumn{2}{|c|}{$ r = 2 $} & \multicolumn{2}{|c|}{$ r = 3 $} & \multicolumn{2}{|c|}{$ r = 4 $} & \multicolumn{2}{|c|}{$ r = 5 $} & \multicolumn{2}{c}{$ r = 6 $} \\[0.3em]
&&&&&&&&\\[-0.8em]
Code $ \mathcal{C}_{\boldsymbol\gamma} $ & $ q^m $ & $ q $ & $ q^m $ & $ q $ & $ q^m $ & $ q $ & $ q^m $ & $ q $ & $ q^m $ & $ q $ \\[0.3em]
\hline\hline
&&&&&&&&\\[-0.8em]
Trivial $ \mathcal{C}_{\boldsymbol\gamma} = \{ 0 \} $ (Lin. RS), $ \mathbf{\forall} h \geq 1 $ & $ \mathbf{2^{10}} $ & $ 2^5 $ & $ \mathbf{2^{15}} $ & $ 2^5 $ & $ \mathbf{2^{20}} $ & $ 2^5 $ & $ \mathbf{2^{20}} $ & $ 2^4 $ & $ \mathbf{2^{24}} $ & $ 2^4 $ \\[0.3em]
\hline 
&&&&&&&&\\[-0.8em]
MDS, $ h = 2 $ & $ \mathbf{2^8} $ & \multirow{3}{*}{$ 2^2 $} & $ 2^{12} $ & \multirow{3}{*}{$ 2^2 $} & $ 2^{16} $ & \multirow{3}{*}{$ 2^2 $} & $ \mathbf{2^{10}} $ &  \multirow{3}{*}{$ 2 $} & $ \mathbf{2^{12}} $ &  \multirow{3}{*}{$ 2 $} \\[0.3em]

&&&&&&&&\\[-0.8em]
\phantom{MDS,} $ h = 3 $ & $ 2^{12} $ & & $ 2^{18} $ & & $ 2^{24} $ & & $ \mathbf{2^{15}} $ & & $ \mathbf{2^{18}} $ & \\[0.3em]

&&&&&&&&\\[-0.8em]
\phantom{MDS,} $ h = 4 $ & $ 2^{16} $ & & $ 2^{24} $ & & $ 2^{32} $ & & $ \mathbf{2^{20}} $ & & $ \mathbf{2^{24}} $ & \\[0.3em]
\hline
&&&&&&&&\\[-0.8em]
Hamming, $ \rho = 3 $, $ h = 2 $ & $ \mathbf{2^8} $ & $ 2 $ & $ \mathbf{2^9} $ & $ 2 $ & $ \mathbf{2^{12}} $ & $ 2 $ & $ 2^{15} $ & $ 2 $ & $ 2^{18} $ & $ 2 $ \\[0.3em]
\hline 
&&&&&&&&\\[-0.8em]
Pr. BCH, $ s = 1,2,3 $, $ h = 2 $ & $ \mathbf{2^8} $ & \multirow{3}{*}{$ 2 $} & $ \mathbf{2^9} $ & \multirow{3}{*}{$ 2 $} & $ \mathbf{2^{12}} $ & \multirow{3}{*}{$ 2 $} & $ \mathbf{2^{10}} $ & \multirow{3}{*}{$ 2 $} & $ \mathbf{2^{12}} $ & \multirow{3}{*}{$ 2 $} \\[0.3em]

&&&&&&&&\\[-0.8em]
\phantom{Pr. BCH, $ s = 1,2,3 $,} $ h = 3 $ & $ 2^{14} $ & & $ \mathbf{2^{15}} $ & & $ \mathbf{2^{20}} $ & & $ \mathbf{2^{15}} $ & & $ \mathbf{2^{18}} $ & \\[0.3em]

&&&&&&&&\\[-0.8em]
\phantom{Pr. BCH, $ s = 1,2,3 $,} $ h = 4 $ & $ 2^{20} $ & & $ 2^{21} $ & & $ 2^{28} $ & & $ \mathbf{2^{20}} $ & & $ \mathbf{2^{24}} $ & \\[0.3em]
\hline\hline
&&&&&&&&\\[-0.8em]
Best MRD code possible, $ \mathbf{\forall} h \geq 1 $ & $ 2^{62} $ & $ 2 $ & $ 2^{62} $ & $ 2 $ & $ 2^{62} $ & $ 2 $ & $ 2^{62} $ & $ 2 $ & $ 2^{62} $ & $ 2 $ \\[0.3em]
\hline 
\end{tabular}
\caption{Table for fixed $ N = gr = 62 $, while other parameters vary. }
\label{table fields MSRD N >= 62}
\end{table}

\end{document}